\documentclass[11pt]{article}

\usepackage[pagebackref,colorlinks]{hyperref}
\usepackage{amsmath} % AMS Math Package
\usepackage{amsthm} % Theorem Formatting
\usepackage{amssymb}	% Math symbols such as \mathbb
\usepackage{graphicx} % Allows for eps images
\usepackage{multicol} % Allows for multiple columns
\usepackage{multirow}
\usepackage{color}
\usepackage{bm}
\usepackage[dvips,letterpaper,margin=1in,bottom=1in]{geometry}
\usepackage[capitalize,noabbrev]{cleveref}

\usepackage[utf8]{inputenc}
\usepackage[english]{babel}

\usepackage{mathtools}

\usepackage{subcaption}

\newtheorem{theorem}{Theorem}[section]

\newtheorem{corollary}[theorem]{Corollary}
\newtheorem{proposition}[theorem]{Proposition}

\newtheorem{definition}{Definition}[section]

\newtheorem{remark}{Remark}[section]
\newtheorem{problem}{Problem}

\newcommand{\braket}[2]{\left< #1 \vphantom{#2} \middle| #2 \vphantom{#1} \right>} % for Dirac brackets
\newcommand{\ketbra}[2]{\ensuremath{\ket{#1}\!\bra{#2}}}

\DeclarePairedDelimiter\rbra{\lparen}{\rparen}
\DeclarePairedDelimiter\sbra{\lbrack}{\rbrack}
\DeclarePairedDelimiter\cbra{\{}{\}}
\DeclarePairedDelimiter\abs{\lvert}{\rvert}
\DeclarePairedDelimiter\Abs{\lVert}{\rVert}

\DeclarePairedDelimiter\ket{\lvert}{\rangle}
\DeclarePairedDelimiter\bra{\langle}{\rvert}

\newcommand{\tr} {\operatorname{tr}}
\newcommand{\poly} {\operatorname{poly}}

\newcommand{\polylog} {\operatorname{polylog}}
\newcommand{\rank} {\operatorname{rank}}

\usepackage{enumerate}
\usepackage{enumitem}

\usepackage{algorithm}
\usepackage{algpseudocode}
\usepackage{tabularx}
\usepackage{booktabs}
\usepackage{threeparttable}

\newcommand{\footremember}[2]{%
    \footnote{#2}
    \newcounter{#1}
    \setcounter{#1}{\value{footnote}}%
}

\usepackage{tikz}
\usetikzlibrary{quantikz2}

\title{Query-Optimal and Sample-Optimal Quantum Algorithms for Estimating Fidelity to a Pure State\footnote{A preliminary version of this paper \cite{FW25c}, where only a query-optimal quantum algorithm was provided, was presented at the 33rd Annual European Symposium on Algorithms (ESA 2025) and also presented as part of the contributed talk \cite{FWZ25} at the 25th Asian Quantum Information Science Conference (AQIS 2025). In this paper, a sample-optimal quantum algorithm is also provided, resolving a question raised in the conference version \cite{FW25c} about the sample complexity of estimating the fidelity to a pure state.}}
\author{
    Wang Fang \footremember{1}{Wang Fang is with the School of Informatics, University of Edinburgh, EH8 9AB Edinburgh, United Kingdom (e-mail: \href{mailto:Wang.Fang@ed.ac.uk}{\nolinkurl{Wang.Fang@ed.ac.uk}}).}
    \and
    Qisheng Wang \footremember{2}{Qisheng Wang is with the School of Computer Science, Shanghai Jiao Tong University, Shanghai 200240, China (e-mail: \href{mailto:QishengWang1994@gmail.com}{\nolinkurl{QishengWang1994@gmail.com}}).}
}
\date{}

\begin{document}

\maketitle

\begin{abstract}
    We present two \textit{optimal} quantum algorithms that estimate the (square root) fidelity of a mixed state to a pure state to within additive error $\varepsilon$: 
    \begin{enumerate}
        \item Given query access to the state-preparation circuits of the input states, the query complexity is shown to be $\Theta(1/\varepsilon)$, achieving a quadratic speedup over the folklore $O(1/\varepsilon^2)$. 
        \item Given sample access to the input states, the sample complexity is shown to be $\Theta(1/\varepsilon^2)$, achieving a quadratic speedup over the folklore $O(1/\varepsilon^4)$. 
    \end{enumerate}
    Our results generalize the previous approaches to pure-state fidelity estimation, and, to the best of our knowledge, are the \textit{first} optimal approaches to fidelity estimation involving mixed states. 
    Our approach is technically simple, and can be extended to estimating the uncommon quantity $\sqrt{\operatorname{tr}(\rho\sigma^2)}$ that is of independent interest. 
\end{abstract}

\textbf{Keywords: quantum computing, fidelity estimation, quantum algorithms, quantum query complexity.}

\newpage

\tableofcontents
\newpage

\section{Introduction}

The fidelity between quantum states \cite{Uhl76,Joz94} is a closeness measure that is commonly used in quantum physics and quantum computing \cite{NC10,Wil13,Hay17,Wat18}. 
Formally, for two (mixed) quantum states $\rho$ and $\sigma$, their (square root) fidelity is defined by (see \cite[Equation (9.53)]{NC10})
\[
    \mathrm{F}\rbra{\rho, \sigma} = \tr\rbra*{\sqrt{\sqrt{\sigma}\rho\sqrt{\sigma}}}.
\]
The fidelity is generally bounded by $0 \leq \mathrm{F}\rbra{\rho, \sigma} \leq 1$. 
In particular, when $\mathrm{F}\rbra{\rho, \sigma} = 1$, the two states $\rho$ and $\sigma$ are identical; and when $\mathrm{F}\rbra{\rho, \sigma} = 0$, the two states are orthogonal. 

The estimation of fidelity turns out to be central in quantum property testing (cf. \cite{MdW16}), which is closely related to quantum hypothesis testing \cite{Kar05,ACM+07,CDL+25}. 
The earliest approach is now known as the SWAP test \cite{BCWdW01}, allowing us to estimate to within additive error $\varepsilon$ the squared fidelity $\mathrm{F}^2\rbra{\rho, \sigma} = \tr\rbra{\rho \sigma}$ when $\rho$ and $\sigma$ are pure states, using $O\rbra{1/\varepsilon^2}$ samples of $\rho$ and $\sigma$. 
Moreover, the squared fidelity $\mathrm{F}^2\rbra{\rho, \sigma}$ for pure states $\rho$ and $\sigma$ can be estimated using $O\rbra{1/\varepsilon}$ queries to their state-preparation circuits through the SWAP test \cite{BCWdW01} equipped with quantum amplitude estimation \cite{BHMT02}. 
The approach based on the SWAP test has been found to have various applications \cite{Nis25}. 
Fidelity estimation for pure states was also considered in some restricted models: a direct fidelity estimation \cite{FL11} was proposed when only Pauli measurements are allowed, and a distributed approach was developed in \cite{ALL22} and further extended in \cite{AS25,GHYZ24,ZWY+25}.
Recently, query-optimal and sample-optimal quantum algorithms for estimating the fidelity $\mathrm{F}\rbra{\rho, \sigma} = \sqrt{\tr\rbra{\rho\sigma}}$ for pure states $\rho$ and $\sigma$ have been found in \cite{Wan24} and \cite{WZ24}, respectively.

Although fidelity estimation is known to be $\mathsf{BQP}$-complete when one of the quantum states is pure \cite{RASW23}, fidelity estimation for mixed states is $\mathsf{QSZK}$-hard in general \cite{Wat02,Wat09}, which means that there is no polynomial-time quantum algorithm for fidelity estimation unless $\mathsf{BQP} = \mathsf{QSZK}$. 
Nevertheless, efficient quantum algorithms for fidelity estimation for low-rank states were recently discovered \cite{WZC+23} and later improved \cite{WGL+24,GP22,UNWT25,LT26,Wan26} with time complexity $\poly\rbra{r, 1/\varepsilon}$, where $r$ is the rank of the states. 
% Remember to add \cite{CW26}
Quantum algorithms for fidelity estimation for well-conditioned states were recently proposed with query complexity $\poly\rbra{\kappa}\cdot \widetilde{\Theta}\rbra{1/\varepsilon}$ \cite{LWWZ25,UNWT25} and with sample complexity $\poly\rbra{\kappa}\cdot \widetilde{\Theta}\rbra{1/\varepsilon^2}$ \cite{UNWT25},\footnote{Throughout this paper, we use $\widetilde{O}\rbra{f} = O\rbra{f \polylog\rbra{f}}$, $\widetilde{\Omega}\rbra{f} = \Omega\rbra{f/\polylog\rbra{f}}$, and $\widetilde{\Theta}\rbra{f} = \widetilde{O}\rbra{f} \cap \widetilde{\Omega}\rbra{f}$.} achieving almost optimal dependence on the precision $\varepsilon$, where $\kappa$ is the condition number such that $\rho, \sigma \geq I/\kappa$. 

In this paper, we present two optimal quantum algorithms for estimating the fidelity of a mixed state $\rho$ to a pure state $\sigma = \ketbra{\psi}{\psi}$: one is for query access and the other is for sample access.
The input models of both quantum algorithms are currently standard models for quantum query complexity and quantum sample complexity, respectively.
The problem of estimating the fidelity of a mixed state to a pure state appears naturally in physics. 
Specifically, an ideal computation or target preparation is described by a pure state, while the state actually produced by a device is mixed because of noise or discarding an environment. 
Estimating the fidelity $\mathrm{F}\rbra{\rho,\ket{\psi}}$ then quantifies how close the actual device output $\rho$ is to the intended pure reference state $\ket{\psi}$ (cf.\ \cite{FL11,dSLCP11}). 
Concrete examples also include verifying quantum devices via fidelity \cite{PLM18,HKP20}.

We compare the complexity of known approaches to fidelity estimation in \cref{tab:cmp}. 
In the low-rank case where either $\rho$ or $\sigma$ has rank at most $r$, prior results showed query complexity upper bounds of $\widetilde{O}\rbra{r^{12.5}/\varepsilon^{13.5}}$ in \cite{WZC+23}, $\widetilde{O}\rbra{r^{6.5}/\varepsilon^{7.5}}$ in \cite{WGL+24}, $\widetilde{O}\rbra{r^{2.5}/\varepsilon^5}$ in \cite{GP22}, and $\widetilde{O}\rbra{r/\varepsilon^2}$ in \cite{UNWT25}, as well as sample complexity upper bounds of $\widetilde{O}\rbra{r^{5.5}/\varepsilon^{12}}$ in \cite{GP22} and $\widetilde{O}\rbra{r^2/\varepsilon^4}$ in \cite{UNWT25}. 
In particular, the fidelity estimation considered in this paper is actually the special case of the low-rank case with $r = 1$, whereas the current best low-rank approach \cite{UNWT25} only imply query and sample upper bounds of $\widetilde{O}\rbra{1/\varepsilon^2}$ and $\widetilde{O}\rbra{1/\varepsilon^4}$, respectively. 
In the well-conditioned case where $\rho, \sigma \geq I/\kappa$ for some parameter $\kappa > 0$, near-optimal query and sample complexity upper bounds $\widetilde{\Theta}\rbra{1/\varepsilon}$ and $\widetilde{\Theta}\rbra{1/\varepsilon^2}$ are obtained in \cite{LWWZ25,UNWT25} when $\kappa = \Theta\rbra{1}$. 
In comparison, those results for well-conditioned fidelity estimation do not apply to the fidelity estimation considered in this paper, as pure states are generally not well-conditioned.

\begin{table}[!htp]
    \centering
    \caption{Complexity of estimating the fidelity $\mathrm{F}\rbra{\rho, \sigma}$.}
    \label{tab:cmp}
    \begin{tabular}{cccc}
        \toprule
        References & Condition & Complexity & Notes \\
        \midrule
        \cite{Wan24} & \multirow{2}{*}{Both $\rho$ and $\sigma$ are pure} & $\Theta\rbra{1/\varepsilon}$ queries & Optimal \\
        \cite{WZ24} & & $\Theta\rbra{1/\varepsilon^2}$ samples & Optimal \\
        \midrule
        \cite{BCWdW01,BHMT02} & \multirow{4}{*}{One of $\rho$ and $\sigma$ is pure} & $O\rbra{1/\varepsilon^2}$ queries & / \\
        \textbf{This Work} & & $\bm{\Theta(1/\varepsilon)}$ \textbf{queries} & \textbf{Optimal} \\
        \cite{BCWdW01} & & $O\rbra{1/\varepsilon^4}$ samples & / \\
        \textbf{This Work} & & $\bm{\Theta(1/\varepsilon^2)}$ \textbf{samples} & \textbf{Optimal} \\
        \midrule
        \cite{WZC+23,WGL+24,GP22,UNWT25} & \multirow{2}{*}{Either $\rho$ or $\sigma$ is of rank $r$} & $\poly\rbra{r, 1/\varepsilon}$ queries & / \\
        \cite{GP22,UNWT25} & & $\poly\rbra{r, 1/\varepsilon}$ samples & / \\
        \midrule
        \cite{LWWZ25,UNWT25} & \multirow{3}{*}{$\rho, \sigma \geq I/\kappa$} & $\poly\rbra{\kappa} \cdot \widetilde{\Theta}\rbra{1/\varepsilon}$ queries & Almost optimal in $\varepsilon$ \\
        \cite{LWWZ25} & & $\poly\rbra{\kappa} \cdot \widetilde{O}\rbra{1/\varepsilon^3}$ samples & / \\ 
        \cite{UNWT25} & & $\poly\rbra{\kappa} \cdot \widetilde{\Theta}\rbra{1/\varepsilon^2}$ samples & Almost optimal in $\varepsilon$ \\
        \bottomrule
    \end{tabular}
\end{table}

\subsection{Query-optimal approach}

In the query-optimal approach, we assume query access to the state-preparation circuits of both the states $\rho$ and $\sigma$.
Specifically, the input model is known as the 
``\textit{purified quantum query access}'' model, where we assume query access to (the controlled version of) a quantum unitary oracle that prepares the purification of the input quantum state, and its inverse. 
This input model is commonly employed in quantum computational complexity \cite{Wat02} and quantum algorithms \cite{GL20}. 
Our query-optimal approach is stated in the following theorem. 

\begin{theorem} [Query-optimal approach to estimating fidelity to a pure state, \cref{thm:fidelity_to_pure} simplified] \label{thm:main}
    Given purified quantum query access to a mixed state $\rho$ and a pure state $\sigma = \ketbra{\psi}{\psi}$, the fidelity $\mathrm{F}\rbra{\rho, \sigma} = \sqrt{\bra{\psi} \rho \ket{\psi}}$ can be estimated to within additive error $\varepsilon$ with query complexity $O\rbra{1/\varepsilon}$.
\end{theorem}

It is worth noting that the approach in \cref{thm:main} further leads to optimal pure-state trace distance estimation in the purified quantum query access model (see \cref{cor:pure-td}).

\paragraph{Techniques.}
The idea is to construct a quantum circuit using the purified quantum query access oracles of $\rho$ and $\sigma$ so that it encodes the fidelity $\mathrm{F}\rbra{\rho, \sigma}$ into the amplitude (see \cref{prop:prob-to-pure}). 
Then, the fidelity $\mathrm{F}\rbra{\rho, \sigma}$ can be estimated by the square root amplitude estimation used in \cite{Wan24}. 

\paragraph{Comparison with previous approaches.}
Prior to the result of \cref{thm:main}, we are only aware of a folklore quantum algorithm based on the SWAP test (which was mentioned above and will be explained in detail in \cref{sec:swap}).
This folklore approach computes $\mathrm{F}^2\rbra{\rho, \sigma} = \bra{\psi} \rho \ket{\psi}$ to within additive error $\varepsilon$ with query complexity $O(1/\varepsilon)$, thereby resulting in a query complexity of $O(1/\varepsilon^2)$ for estimating $\mathrm{F}\rbra{\rho, \ket{\psi}}$ to within additive error $\varepsilon$. 
By comparison, \cref{thm:main} exhibits a quadratic improvement.

\paragraph{Generalization of the pure-state fidelity (and trace distance) estimation in \cite{Wan24}.}
The quantum query algorithm given in \cref{thm:main} generalizes the pure-state fidelity (and trace distance) estimation in \cite{Wan24} where unitary oracles are required to directly prepare the pure states. 
In comparison, our result in \cref{thm:main} allows the unitary oracles to have redundant output qubits (see \cref{sec:imp} for more details). 
The advantage is that \cref{thm:main} works for input models that are more general. 
Moreover, this advantage further leads to a sample-optimal approach for estimating the fidelity to a pure state (see \cref{sec:sample-optimal-approach}), whereas the pure-state fidelity estimation in \cite{Wan24} does not apply (see \cref{remark:sample}). 

\paragraph{Query-optimality.}

Our approach is actually optimal, as it can be used to estimate the fidelity between two pure states 
and the quantum query lower bound $\Omega\rbra{1/\varepsilon}$ is due to \cite{BBC+01,NW99} as noted in \cite{Wan24} (see \cref{thm:lower-bound-pure}\ref{item:pure-query}). 
In \cref{thm:lower-bound-rank}\ref{item:query}, we further show that the query-optimality holds even if $\rho$ is of an arbitrary rank. 

\subsection{Sample-optimal approach} \label{sec:sample-optimal-approach}

In the sample-optimal approach, we assume sample access to identical copies of both the states $\rho$ and $\sigma$. 
The input model is the standard ``\textit{sample access}'' model in quantum computing, and it is commonly employed in quantum state tomography \cite{HHJ+17,OW16,SSW25,PSW25,PSTW25} and property testing of quantum states \cite{CHW07,OW21,AISW20,BOW19,Hay25,CW25}. 
Our sample-optimal approach is stated in the following theorem.

\begin{theorem}[Sample-optimal approach to estimating fidelity to a pure state, \cref{thm:est-fid-sampl} simplified] \label{thm:main-s}
    Given sample access to a mixed state $\rho$ and a pure state $\sigma = \ketbra{\psi}{\psi}$, the fidelity $\mathrm{F}\rbra{\rho, \sigma} = \sqrt{\bra{\psi} \rho \ket{\psi}}$ can be estimated to within additive error $\varepsilon$ with sample complexity $O\rbra{1/\varepsilon^2}$. 
\end{theorem}

\paragraph{Techniques.}
This is achieved by the strengthened version \cite{TWZ25} of quantum sample-to-query lifting \cite{WZ25a,WZ25b} (see also \cite{CWZ25}), which can convert a quantum algorithm with query complexity $Q$ in the purified quantum query access model to another quantum algorithm with sample complexity $O\rbra{Q^2}$. 
Given this, the query complexity $O\rbra{1/\varepsilon}$ in \cref{thm:main} immediately gives the sample complexity $O\rbra{1/\varepsilon^2}$ in \cref{thm:main-s}. 
It should be noted that the need of purified quantum query access required in \cref{thm:main} is necessary to establish the result in \cref{thm:main-s}  (see \cref{remark:sample}). 

\paragraph{Comparison with previous approaches.}
Prior to the result of \cref{thm:main-s}, we are only aware of a folklore quantum algorithm based on the SWAP test (see \cref{sec:swap}).
This folklore approach can compute $\mathrm{F}^2\rbra{\rho, \sigma} = \bra{\psi} \rho \ket{\psi}$ to within additive error $\varepsilon$ with sample complexity $O(1/\varepsilon^2)$, thereby resulting in a sample complexity of $O(1/\varepsilon^4)$ for estimating $\mathrm{F}\rbra{\rho, \ket{\psi}}$ to within additive error $\varepsilon$. 
By comparison, \cref{thm:main-s} exhibits a quadratic improvement.

\paragraph{Comparison with the pure-state fidelity estimation in \cite{WZ24}.}

The sample complexity $O\rbra{1/\varepsilon^2}$ for estimating the fidelity $\mathrm{F}\rbra{\rho, \ket{\psi}}$ in \cref{thm:main-s} reproduces the result in \cite{WZ24} that the sample complexity for estimating the fidelity $\mathrm{F}\rbra{\ket{\varphi}, \ket{\psi}}$ between two pure states $\ket{\varphi}$ and $\ket{\psi}$ is $O\rbra{1/\varepsilon^2}$. 
    In comparison, 
    \begin{enumerate}
        \item \cref{thm:main-s} applies to a more general case where one quantum state is mixed, whereas the result in \cite{WZ24} only applies to the case where both quantum states are pure,
        \item The approach in \cite{WZ24} has time complexity (i.e., the number of elementary quantum gates) $O\rbra{k/\varepsilon^2}$ when the quantum states are $k$-qubit, whereas the time complexity $\widetilde{O}\rbra{\poly\rbra{k}/\varepsilon^8}$ of the approach in \cref{thm:main-s} is polynomially worse (see \cref{thm:est-fid-sampl} for the time complexity of \cref{thm:main-s}). 
    \end{enumerate}

\paragraph{Sample-optimality.}

Our approach is actually optimal, as it can be used to estimate the fidelity between two pure states 
and the quantum sample lower bound $\Omega\rbra{1/\varepsilon^2}$ is due to \cite{ALL22} (see \cref{thm:lower-bound-pure}\ref{item:pure-sample}). 
In \cref{thm:lower-bound-rank}\ref{item:sample}, we further show that the sample-optimality holds even if $\rho$ is of an arbitrary rank. 

\subsection{Extensions}

As a bonus, our approaches in \cref{thm:main,thm:main-s} can further estimate the quantity $\sqrt{\tr\rbra{\rho\sigma^2}}$ with optimal query complexity and sample complexity (see \cref{thm:trace_rho_sigma2,thm:trace_rho_sigma2-sample}), which is of independent interest as this quantity is not common in the literature. 
Nevertheless, we hope this will provide new insights into the development of quantum computing.

\subsection{Organization of this paper}

We give a formal definition of the purified quantum query access model and the problem statement of fidelity estimation in \cref{sec:def}.
Then, we review previous approaches in \cref{sec:warmup} with their subroutines. 
In \cref{sec:algo}, we present our query-optimal approach and its generalization, with an implication discussed in \cref{sec:imp}. 
In \cref{sec:sample}, we present our sample-optimal approach and its generalization. 
Lower bounds on both the query and sample complexities for fidelity estimation are discussed in \cref{sec:lb}. 
Finally, a brief discussion is drawn in \cref{sec:discussion}. 

\section{Problem Settings} \label{sec:def}

In this paper, we assume purified quantum query access to the input quantum states. 
This input model is widely used in the literature, e.g., \cite{GL20,SH21,GHS21,GP22,WZL24,WZ24b,LW25}. 

\begin{definition} [Purified quantum query access] \label{def:query}
    For a mixed quantum state $\rho$, purified quantum query access to $\rho$ means query access to a unitary oracle $U$ that prepares a purification of $\rho$.
    Specifically, suppose $U_{\mathsf{A}\mathsf{B}}$ acts on two subsystems $\mathsf{A}$ and $\mathsf{B}$, then $\rho_{\mathsf{A}} = \tr_{\mathsf{B}}\rbra{\ketbra{\rho}{\rho}_{\mathsf{A}\mathsf{B}}}$, where $\ket{\rho}_{\mathsf{A}\mathsf{B}} = U_{\mathsf{A}\mathsf{B}} \ket{0}_{\mathsf{A}\mathsf{B}}$.
    Moreover, we are allowed to use queries to (controlled-)$U$ and its inverse. 
\end{definition}

We consider the problem of query-access fidelity estimation, formally stated as follows. 

\begin{problem} [Query-access fidelity estimation] \label{prob:fidelity}
    Given purified quantum query access to two mixed quantum states $\rho$ and $\sigma$, the task is to compute $\mathrm{F}\rbra{\rho, \sigma}$ to within additive error $\varepsilon$. 
\end{problem}

We also consider the problem of sample-access fidelity estimation, formally stated as follows.

\begin{problem}[Sample-access fidelity estimation] \label{prob:fidelity-s}
    Given sample access to (i.e., sufficiently many identical copies of) two mixed quantum states $\rho$ and $\sigma$, the task is to compute $\mathrm{F}\rbra{\rho, \sigma}$ to within additive error $\varepsilon$. 
\end{problem}

In particular, this paper focuses on the case of Problems \ref{prob:fidelity} and \ref{prob:fidelity-s} where $\sigma = \ketbra{\psi}{\psi}$ is pure. 

\section{Warm-Ups} \label{sec:warmup}

As a warm-up, we review some previous approaches to fidelity estimation involving pure states, together with their useful subroutines, that are comparable to and inspire our results. 

\subsection{SWAP test} \label{sec:swap}

It is well known that the SWAP test \cite{BCWdW01} can be used to estimate the fidelity $\mathrm{F}\rbra{\rho, \sigma}$ when one of $\rho$ and $\sigma$ is pure. 
In particular, if $\sigma = \ketbra{\psi}{\psi}$ is pure, the SWAP test on input $\rho$ and $\sigma$ (see \cref{fig:swap}) outputs a bit $x \in \cbra{0, 1}$ such that (adapted from \cite[Proposition 9]{KMY09})
\begin{equation} \label{eq:swap-prob}
    \Pr\sbra{x = 0} = \frac{1 + \mathrm{F}^2\rbra{\rho, \ket{\psi}}}{2}.
\end{equation}
With this, we can estimate $\mathrm{F}\rbra{\rho, \ket{\psi}}$.
Specifically, suppose that $\tilde p$ is an estimate of $\Pr\sbra{x = 0}$ such that $\abs{\tilde p - \Pr\sbra{x = 0}} \leq \delta$. 
Then, it can be shown that $\sqrt{2\tilde p - 1}$ is an estimate of $\mathrm{F}\rbra{\rho, \ket{\psi}}$ with $\abs{\sqrt{2\tilde p - 1} - \mathrm{F}\rbra{\rho, \ket{\psi}}} \leq \Theta\rbra{\sqrt{\delta}}$. 
An estimate of $\mathrm{F}\rbra{\rho, \ket{\psi}}$ to within additive error $\varepsilon$ can be obtained by setting $\delta = \Theta\rbra{\varepsilon^2}$. 

\begin{figure} [!htp]
\centering
\begin{quantikz} [row sep = {20pt, between origins}]
    \lstick{$\ket{0}$} & \gate{H} & \ctrl{2} & \gate{H} & \meter{} & \setwiretype{c} & \rstick{$x \in \cbra{0, 1}$} \\
    \lstick{$\rho$~} & \qw & \swap{1} & \qw & \qw \\
    \lstick{$\ket{\psi}$} & \qw & \targX{} & \qw & \qw \\
\end{quantikz}
\caption{The SWAP test for estimating $\mathrm{F}\rbra{\rho, \ket{\psi}}$.}
\label{fig:swap}
\end{figure}

To estimate $\mathrm{F}\rbra{\rho, \ket{\psi}}$ when given purified quantum query access to $\rho$ and $\ket{\psi}$, we need the quantum subroutine for amplitude estimation \cite{BHMT02}. 

\begin{theorem} [Quantum amplitude estimation, \cite{BHMT02}] \label{thm:ampl-esti}
    Suppose that $U$ is a unitary operator such that $U\ket{0}_\mathsf{A}\ket{0}_{\mathsf{B}} = \sqrt{p} \ket{0}_{\mathsf{A}} \ket{\varphi_0}_{\mathsf{B}} + \sqrt{1-p} \ket{1}_{\mathsf{A}} \ket{\varphi_1}_{\mathsf{B}}$, where $p \in \sbra{0, 1}$ and $\ket{\varphi_0}, \ket{\varphi_1}$ are normalized pure states. 
    Then, we can estimate $p$ to within additive error $\delta$ using $O\rbra{1/\delta}$ queries to (controlled-)$U$ and its inverse. 
\end{theorem}

By \cref{thm:ampl-esti}, we can therefore obtain an estimate of $\mathrm{F}\rbra{\rho, \ket{\psi}}$ to within additive error $\varepsilon$ using $O\rbra{1/\varepsilon^2}$ queries to the state-preparation circuits of $\rho$ and $\ket{\psi}$, based on the SWAP test in \cref{fig:swap}. 

\subsection{Pure-state fidelity estimation} \label{sec:pure-fidelity}

Recently, a new approach was proposed in \cite{Wan24} for estimating the fidelity $\mathrm{F}\rbra{\ket{\varphi}, \ket{\psi}}$ between two pure states $\ket{\varphi}$ and $\ket{\psi}$, improving the folklore approach in \cref{sec:swap}. 
However, the approach in \cite{Wan24} assumes a slightly more restricted input model, where $U_{{\varphi}}$ and $U_{{\psi}}$ are two unitary quantum circuits that respectively prepare $\ket{\varphi}$ and $\ket{\psi}$, i.e., $U_{{\varphi}} \ket{0} = \ket{\varphi}$ and $U_{{\psi}} \ket{0} = \ket{\psi}$.

The key idea of \cite{Wan24} is to \textit{encode} the value of $\mathrm{F}\rbra{\ket{\varphi}, \ket{\psi}}$ in the amplitudes of a quantum state, rather than the $\sqrt{\rbra{1 + \mathrm{F}^2\rbra{\ket{\varphi}, \ket{\psi}}}/{2}}$ (in \cref{eq:swap-prob}). 
Specifically, $\mathrm{F}\rbra{\ket{\varphi}, \ket{\psi}}$ can be encoded by the unitary operator $W = U_{\varphi}^\dag U_{\psi}$ such that
\begin{equation} \label{eq:def-W}
    W \ket{0} = \mathrm{F}\rbra{\ket{\varphi}, \ket{\psi}} \ket{0} + \ket{\perp},
\end{equation}
where $\ket{\perp}$ is an (unnormalized) pure state that is orthogonal to $\ket{0}$. 
The circuit $W$ is depicted in \cref{fig:qc-w} and it outputs a bit $x \in \cbra{0, 1}$ such that 
\begin{equation} \label{eq:prob-pure-state-fidelity-estimation}
    \Pr\sbra{x = 0} = \mathrm{F}^2\rbra{\ket{\varphi}, \ket{\psi}}.
\end{equation}

\begin{figure} [!htp]
\centering
\begin{quantikz} [row sep = {20pt, between origins}]
    \lstick{$\ket{0}$} & \gate[][20][20]{U_{\psi}} & \gate[][20][20]{U_{\varphi}^\dag} & \meter{} & \setwiretype{c} & \rstick{$x \in \cbra{0, 1}$}
\end{quantikz}
\caption{The quantum circuit for estimating $\mathrm{F}\rbra{\ket{\varphi}, \ket{\psi}}$.}
\label{fig:qc-w}
\end{figure}

Finally, the value of $\mathrm{F}\rbra{\ket{\varphi}, \ket{\psi}}$ can be estimated by the square root amplitude estimation provided in \cite{Wan24}. 

\begin{theorem} [Quantum square root amplitude estimation, {\cite[Theorem III.4]{Wan24}}] \label{thm:sqrt}
    Suppose that $U$ is a unitary operator such that $U\ket{0}_\mathsf{A}\ket{0}_{\mathsf{B}} = \sqrt{p} \ket{0}_{\mathsf{A}} \ket{\varphi_0}_{\mathsf{B}} + \sqrt{1-p} \ket{1}_{\mathsf{A}} \ket{\varphi_1}_{\mathsf{B}}$, where $p \in \sbra{0, 1}$ and $\ket{\varphi_0}, \ket{\varphi_1}$ are normalized pure states. 
    Then, we can estimate $\sqrt{p}$ to within additive error $\delta$ using $O\rbra{1/\delta}$ queries to (controlled-)$U$ and its inverse and $O\rbra{\log\rbra{1/\delta}\log\log\rbra{1/\delta}}$ additional elementary quantum gates.\footnote{In \cite[Theorem III.4]{Wan24}, the quantum square root amplitude estimation was essentially implemented through quantum phase estimation \cite{Kit95,Sho97}.
    In the textbook approach \cite{NC10} of quantum phase estimation, the number of additional elementary quantum gates is $O\rbra{\log^2\rbra{1/\delta}}$ due to the use of the textbook quantum Fourier transform. This can be further improved to $O\rbra{\log\rbra{1/\delta}\log\log\rbra{1/\delta}}$ using the improved quantum Fourier transform \cite{HH00}.} 

    For convenience, throughout this paper, we use $\mathsf{SqrtAmpEst}\rbra{U, \varepsilon}$ to denote (the returned value of) the square root amplitude estimation process. 
\end{theorem}

Compared to \cref{thm:ampl-esti}, \cref{thm:sqrt} gives an estimate of $\sqrt{p}$ (instead of $p$), thereby allowing us to skip taking square roots. 
By \cref{thm:sqrt}, we can therefore obtain an estimate of $\mathrm{F}\rbra{\ket{\varphi}, \ket{\psi}}$ to within additive error $\varepsilon$ using $O\rbra{1/\varepsilon}$ queries to the state-preparation circuits of $\ket{\varphi}$ and $\ket{\psi}$, based on the construction of $W$ in \cref{eq:def-W}.

\section{Query-Optimal Approach} \label{sec:algo}

For the task of estimating the fidelity $\mathrm{F}\rbra{\rho, \ket{\psi}}$ of a mixed state to a pure state, the approach in \cref{sec:swap} only gives a query complexity of $O\rbra{1/\varepsilon^2}$, while the approach in \cref{sec:pure-fidelity}, however, turns out not to be applicable (as a more restricted input model is assumed). 
In this section, we present a simple quantum algorithm that estimates $\mathrm{F}\rbra{\rho, \ket{\psi}}$ to within additive error $\varepsilon$ using $O\rbra{1/\varepsilon}$ queries to the state-preparation circuits of $\rho$ and $\ket{\psi}$. 

Suppose that the purified quantum query access to $\rho$ and $\ket{\psi}$ is given by two quantum circuits $U$ and $V$, respectively:
\begin{align}
    U_{\mathsf{AB}} \ket{0}_{\mathsf{A}} \ket{0}_{\mathsf{B}} & = \ket{\rho}_{\mathsf{AB}}, \label{eq:purification_rho} \\
    V_{\mathsf{A'B'}} \ket{0}_{\mathsf{A'}} \ket{0}_{\mathsf{B'}} & = \ket{\psi}_{\mathsf{A'}} \ket{\psi'}_{\mathsf{B'}}, \label{eq:purification_psi}
\end{align}
where $\rho_{\mathsf{A}} = \tr_{\mathsf{B}}\rbra{\ketbra{\rho}{\rho}_{\mathsf{AB}}}$ and $\ket{\psi'}_{\mathsf{B'}}$ is any pure state. 
Here, $\mathsf{A}$, $\mathsf{B}$, $\mathsf{A'}$, and $\mathsf{B'}$ are subscripts of subsystems for clarity, where the subsystems $\mathsf{A}$ and $\mathsf{A'}$ have the same dimension. 
Without loss of generality, we can also assume that the subsystems $\mathsf{B}$ and $\mathsf{B'}$ have the same dimension.\footnote{If the subsystem $\mathsf{B}$ has a larger dimension than $\mathsf{B'}$, then $V_{\mathsf{A'B'}} \otimes I_{\mathsf{C}}$ (for certain subsystem $\mathsf{C}$) can be a purified quantum query oracle for $\ket{\psi}$ with the ancilla system $\mathsf{B'C}$ having the same dimension as $\mathsf{B}$. A similar construction also applies to the case that the subsystem $\mathsf{B'}$ has a larger dimension than $\mathsf{B}$.} 

\subsection{Estimating fidelity to a pure state}

Our idea is to encode $\mathrm{F}\rbra{\rho, \ket{\psi}}$ in the amplitudes of an efficiently preparable quantum state, and then take use of the square root amplitude estimation in \cref{thm:sqrt}.
To this end, we present a quantum circuit
\begin{equation}\label{eq:fidelity_w}
    W = \rbra{V_{\mathsf{AB}}^\dag\otimes I_{\mathsf{A'B'}}}\cdot \mathsf{SWAP}_{\mathsf{BB'}} \cdot \rbra{U_{\mathsf{AB}}\otimes V_{\mathsf{A'B'}}}.
\end{equation}
This circuit is depicted in \cref{fig:q-von-ex} and it outputs two bits $x_{\mathsf{A}}, x_{\mathsf{B}} \in \cbra{0, 1}$ such that 
\begin{equation} \label{eq:prob-to-pure}
    \Pr\sbra{x_{\mathsf{A}} = x_{\mathsf{B}} = 0} = \mathrm{F}^2\rbra{\rho, \ket{\psi}}. 
\end{equation}
\cref{eq:prob-to-pure} is comparable to \cref{eq:swap-prob} for the SWAP test and \cref{eq:prob-pure-state-fidelity-estimation} for the pure-state fidelity estimation, and it can be verified by the following proposition. 
\begin{figure} [htp]
\centering
\begin{quantikz}
    \lstick{$\ket{0}_{\mathsf{A}}$} & \gate[2]{U} \gategroup[wires=4,steps=3,style={dashed}]{$W$} & \permute{1,4,3,2} & \gate[2]{V^\dag} & \meter{} & \setwiretype{c} \rstick{$x_\mathsf{A}$} \\
    \lstick{$\ket{0}_{\mathsf{B}}$} & & & & \meter{} & \setwiretype{c} \rstick{$x_\mathsf{B}$} \\
    \lstick{$\ket{0}_{\mathsf{A'}}$} & \gate[2]{V} & & & \\
    \lstick{$\ket{0}_{\mathsf{B'}}$} & & & &
\end{quantikz}
\caption{The encoding unitary operator $W$ for $\mathrm{F}\rbra{\rho, \ket{\psi}}$.}
\label{fig:q-von-ex}
\end{figure}
\begin{proposition} \label{prop:prob-to-pure}
    $\Abs*{\rbra*{\bra{0}_{\mathsf{A}} \bra{0}_{\mathsf{B}} \otimes I_{\mathsf{A'B'}}}W\ket{0}_\mathsf{A}\ket{0}_\mathsf{B}\ket{0}_\mathsf{A'}\ket{0}_\mathsf{B'}}^2 = \mathrm{F}^2\rbra{\rho, \ket{\psi}}$, where $W$ is defined by \cref{eq:fidelity_w}, $U_{\mathsf{AB}}$ is defined by \cref{eq:purification_rho}, and $V_{\mathsf{A}'\mathsf{B}'}$ is defined by \cref{eq:purification_psi}. 
\end{proposition}
\begin{proof}
    This can be shown by direct calculations. 
    \begin{align*}
    & \Abs*{\rbra*{\bra{0}_\mathsf{A}\bra{0}_{\mathsf{B}}\otimes I_{\mathsf{A'B'}}}W\ket{0}_\mathsf{A}\ket{0}_\mathsf{B}\ket{0}_\mathsf{A'}\ket{0}_\mathsf{B'}}^2 \\
    ={}& \Abs*{\rbra*{\bra{0}_\mathsf{A}\bra{0}_{\mathsf{B}}\otimes I_{\mathsf{A'B'}}} \rbra{V_{\mathsf{AB}}^\dag\otimes I_{\mathsf{A'B'}}}\cdot \mathsf{SWAP}_{\mathsf{BB'}} \cdot \rbra{U_{\mathsf{AB}}\otimes V_{\mathsf{A'B'}}} \ket{0}_\mathsf{A}\ket{0}_\mathsf{B}\ket{0}_\mathsf{A'}\ket{0}_\mathsf{B'}}^2 \\
    ={}& \Abs*{\rbra*{\bra{\psi}_\mathsf{A}\bra{\psi'}_{\mathsf{B}}\otimes I_{\mathsf{A'B'}}} \cdot \mathsf{SWAP}_{\mathsf{BB'}} \cdot \ket{\rho}_\mathsf{AB}\ket{\psi}_\mathsf{A'}\ket{\psi'}_\mathsf{B'}}^2 \\
    ={}& \Abs*{\rbra*{\bra{\psi}_\mathsf{A}\otimes I_{\mathsf{BA'}}\otimes\bra{\psi'}_{\mathsf{B'}}} \cdot \ket{\rho}_\mathsf{AB}\ket{\psi}_\mathsf{A'}\ket{\psi'}_\mathsf{B'}}^2 \\
    ={}& \Abs*{\rbra*{\bra{\psi}_\mathsf{A}\otimes I_{\mathsf{BA'}}} \cdot \ket{\rho}_\mathsf{AB}\ket{\psi}_\mathsf{A'}}^2 \\
    ={}& \bra{\rho}_{\mathsf{AB}}\bra{\psi}_{\mathsf{A'}} \cdot \rbra*{\ketbra{\psi}{\psi}_{\mathsf{A}} \otimes I_\mathsf{BA'}} \cdot \ket{\rho}_\mathsf{AB}\ket{\psi}_\mathsf{A'} \\
    ={}& \bra{\rho}_{\mathsf{AB}} \cdot \rbra*{\ketbra{\psi}{\psi}_{\mathsf{A}} \otimes I_\mathsf{B}} \cdot \ket{\rho}_\mathsf{AB} \\
    ={}& \tr\rbra[\big]{\rbra*{\ketbra{\psi}{\psi}_{\mathsf{A}} \otimes I_\mathsf{B}} \cdot \ketbra{\rho}{\rho}_{\mathsf{AB}}} \\
    ={}& \tr\rbra[\Big]{\tr_\mathsf{B}\rbra[\big]{\rbra*{\ketbra{\psi}{\psi}_{\mathsf{A}} \otimes I_\mathsf{B}} \cdot \ketbra{\rho}{\rho}_{\mathsf{AB}}}}\\
    ={}& \tr\rbra[\big]{\ketbra{\psi}{\psi}_{\mathsf{A}} \cdot \tr_\mathsf{B}\rbra*{\ketbra{\rho}{\rho}_{\mathsf{AB}}}} \\
    ={}& \tr\rbra*{\ketbra{\psi}{\psi}_{\mathsf{A}} \cdot \rho_\mathsf{A}} \\
    ={}& \mathrm{F}^2\rbra{\rho, \ket{\psi}}.
\end{align*}
\end{proof}

According to \cref{prop:prob-to-pure}, we can write 
\begin{equation} \label{eq:W-amp}
    W\ket{0}_\mathsf{A}\ket{0}_\mathsf{B}\ket{0}_\mathsf{A'}\ket{0}_\mathsf{B'} = \mathrm{F}\rbra{\rho,\ket{\psi}}\ket{0}_{\mathsf{A}}\ket{0}_{\mathsf{B}}\ket{\phi}_{\mathsf{A'B'}}+\sqrt{1-\mathrm{F}^2\rbra{\rho, \ket{\psi}}}\ket{\phi^{\perp}}_{\mathsf{ABA'B'}}
\end{equation}
for some normalized pure states $\ket{\phi}$ and $\ket{\phi^{\perp}}$ such that $\rbra*{\ketbra{0}{0}_{\mathsf{A}}\otimes\ketbra{0}{0}_{\mathsf{B}}\otimes I_{\mathsf{A'B'}}}\ket{\phi^{\perp}} = 0$.
With this, we can therefore estimate $\mathrm{F}\rbra{\rho, \ket{\psi}}$ using the square root amplitude estimation (\cref{thm:sqrt}). 
We formally state the result and give its rigorous proof as follows. 

\begin{algorithm*}[ht]
    \caption{Quantum query algorithm for estimating fidelity to a pure state.}
    \label{algo:fidelity_to_pure}
    \begin{algorithmic}[1]
        \Require Quantum oracles $U$ and $V$ that prepare $n$-qubit purifications of a $k$-qubit mixed state $\rho$ and a $k$-qubit pure state $\ketbra{\psi}{\psi}$, respectively (as well as $U^\dag$, $V^\dag$, and their controlled versions); the desired additive error $\varepsilon \in \rbra{0, 1}$. 

        \Ensure An estimate of $\mathrm{F}\rbra{\rho, \ket{\psi}}$ to within additive error $\varepsilon$ with probability at least $2/3$.
        
        \State Let unitary operator 
        \begin{align*}
            W' &= \rbra*{I_\mathsf{C}\otimes \ketbra{0}{0}_{\mathsf{A}}\otimes\ketbra{0}{0}_{\mathsf{B}}\otimes I_{\mathsf{A'B'}}+X_{\mathsf{C}}\otimes \rbra*{I_{\mathsf{AB}} - \ketbra{0}{0}_{\mathsf{A}}\otimes\ketbra{0}{0}_{\mathsf{B}}}\otimes I_{\mathsf{A'B'}}} \\
            & \qquad \cdot  \rbra{V_{\mathsf{AB}}^\dag\otimes I_{\mathsf{A'B'}}}\cdot \mathsf{SWAP}_{\mathsf{BB'}} \cdot \rbra{U_{\mathsf{AB}}\otimes V_{\mathsf{A'B'}}}.
        \end{align*}

        \State \Return $\mathsf{SqrtAmpEst}\rbra{W', \varepsilon}$ by \cref{thm:sqrt}.
    \end{algorithmic}
\end{algorithm*}

\begin{theorem}[Estimating fidelity to a pure state given purified quantum query access] \label{thm:fidelity_to_pure}
    Suppose that $U$ and $V$ are quantum unitary operators that prepare $n$-qubit purifications of a $k$-qubit mixed state $\rho = \tr_\mathsf{B}\rbra{\ketbra{\rho}{\rho}_{\mathsf{AB}}}$ and a $k$-qubit pure state $\ketbra{\psi}{\psi} = \ketbra{\psi}{\psi}_{\mathsf{A'}}$, respectively, as in \cref{eq:purification_rho,eq:purification_psi}.
    For $\varepsilon \in \rbra{0,1}$, there is a quantum query algorithm that estimates $\mathrm{F}\rbra{\rho, \ket{\psi}}$ to within additive error $\varepsilon$ with probability at least $2/3$ using $O\rbra{1/\varepsilon}$ queries to (controlled-)$U$, (controlled-)$V$, and their inverses and $O\rbra{n/\varepsilon}$ elementary quantum gates in addition to oracle queries.
\end{theorem}
\begin{proof}
The formal algorithm is given in \cref{algo:fidelity_to_pure}. 
To make the proof rigorous, we introduce an auxiliary system $\mathsf{C}$ that consists of one qubit.
Let
\begin{equation}\label{eq:fidelity_w'}
    W' = \rbra*{\rbra[\big]{I_\mathsf{C}\otimes \ketbra{0}{0}_{\mathsf{A}}\otimes\ketbra{0}{0}_{\mathsf{B}}+X_{\mathsf{C}}\otimes \rbra*{I_{\mathsf{AB}} - \ketbra{0}{0}_{\mathsf{A}}\otimes\ketbra{0}{0}_{\mathsf{B}}}}\otimes I_{\mathsf{A'B'}}} W,
\end{equation}
where $W$ is defined by \cref{eq:fidelity_w}.
According to \cref{eq:W-amp}, we have
\[
    W'\ket{0}_\mathsf{C}\ket{0}_\mathsf{A}\ket{0}_\mathsf{B}\ket{0}_\mathsf{A'}\ket{0}_\mathsf{B'} = 
    \mathrm{F}\rbra{\ket{\psi}, \rho}\ket{0}_{\mathsf{C}}\ket{0}_{\mathsf{A}}\ket{0}_{\mathsf{B}}\ket{\phi}_{\mathsf{A'B'}}+\sqrt{1-\mathrm{F}\rbra{\ket{\psi}, \rho}^2}\ket{1}_{\mathsf{C}}\ket{\phi^{\perp}}_{\mathsf{ABA'B'}}.
\]
Finally, we can obtain an estimate $\tilde{x} = \mathsf{SqrtAmpEst}\rbra{W', \varepsilon}$ of $\mathrm{F}\rbra{\rho, \ket{\psi}}$ to within additive error $\varepsilon$ (i.e., $\abs*{\tilde{x} - \mathrm{F}\rbra{\rho, \ket{\psi}}} < \varepsilon$) with probability at least $2/3$ using $O\rbra{1/\varepsilon}$ queries to $W'$.
According to \cref{eq:fidelity_w',eq:fidelity_w}, one query to $W'$ consists of one query to $U$, two queries to $V$, and $O(n)$ elementary quantum gates (the SWAP operation on $\mathsf{BB'}$ and the multi-controlled-$X$ gate on $\mathsf{CAB}$) in addition to oracle queries.
Therefore, the way we obtain $\tilde{x}$ uses $O\rbra{1/\varepsilon}$ queries to both $U$ and $V$ and $O\rbra{n/\varepsilon}$ elementary quantum gates in addition to oracle queries by \cref{thm:sqrt}.
\end{proof}

\subsection{Generalizations} \label{sec:general}

In this subsection, we naturally extend the unitary operator in \cref{fig:q-von-ex} to the scenario where $V$ prepares a purification of an arbitrary mixed state $\sigma$.
That is, we replace \cref{eq:purification_psi} with
\begin{equation} 
    V_{\mathsf{A'B'}}\ket{0}_\mathsf{A'}\ket{0}_\mathsf{B'} = \ket{\sigma}_{\mathsf{A'B'}}, \label{eq:purification_sigma}
\end{equation}
where $\sigma_{\mathsf{A'}} = \tr_{\mathsf{B'}}\rbra{\ketbra{\sigma}{\sigma}_{\mathsf{A'B'}}}$.
In this case, it turns out that $\sqrt{\tr\rbra{\rho\sigma^2}}$ is encoded in the amplitude.
Note that when $\sigma$ is pure, $\sigma^2 = \sigma$, the quantity $\sqrt{\tr\rbra{\rho\sigma^2}}$ equals the fidelity $\mathrm{F}\rbra{\rho, \sigma}$. 
Specifically, the circuit in \cref{fig:q-von-ex} outputs two bits $x_{\mathsf{A}}, x_{\mathsf{B}} \in \cbra{0, 1}$ such that 
\begin{equation}
    \Pr\sbra{x_{\mathsf{A}} = x_{\mathsf{B}} = 0} = \tr\rbra{\rho\sigma^2},
\end{equation}
which generalizes \cref{eq:prob-to-pure}. 
This can be verified by the following proposition. 

\begin{proposition} \label{prop:prob-mixed}
    $\Abs*{\rbra*{\bra{0}_{\mathsf{A}} \bra{0}_{\mathsf{B}} \otimes I_{\mathsf{A'B'}}}W\ket{0}_\mathsf{A}\ket{0}_\mathsf{B}\ket{0}_\mathsf{A'}\ket{0}_\mathsf{B'}}^2 = \tr\rbra{\rho\sigma^2}$, where $W$ is defined by \cref{eq:fidelity_w}, $U_{\mathsf{AB}}$ is defined by \cref{eq:purification_rho}, and $V_{\mathsf{A}'\mathsf{B}'}$ is defined by \cref{eq:purification_sigma}. 
\end{proposition}

\begin{proof}
    This can be shown by direct calculations generalizing the proof of \cref{prop:prob-to-pure}.
\begin{align*}
    & \Abs*{\rbra*{\bra{0}_\mathsf{A}\bra{0}_{\mathsf{B}}\otimes I_{\mathsf{A'B'}}}W\ket{0}_\mathsf{A}\ket{0}_\mathsf{B}\ket{0}_\mathsf{A'}\ket{0}_\mathsf{B'}}^2 \\
    ={}& \Abs*{\rbra*{\bra{\sigma}_{\mathsf{AB}}\otimes I_{\mathsf{A'B'}}} \cdot \mathsf{SWAP}_{\mathsf{BB'}} \cdot \ket{\rho}_{\mathsf{AB}}\ket{\sigma}_{\mathsf{A'B'}}}^2 \\
    ={}& \bra{\rho}_{\mathsf{AB}}\bra{\sigma}_{\mathsf{A'B'}} \cdot \mathsf{SWAP}_{\mathsf{BB'}} \cdot \rbra*{\ketbra{\sigma}{\sigma}_{\mathsf{AB}}\otimes I_{\mathsf{A'B'}}} \cdot \mathsf{SWAP}_{\mathsf{BB'}} \cdot \ket{\rho}_{\mathsf{AB}}\ket{\sigma}_{\mathsf{A'B'}} \\
    ={}& \tr\rbra[\big]{\mathsf{SWAP}_{\mathsf{BB'}}\rbra*{\ketbra{\sigma}{\sigma}_{\mathsf{AB}}\otimes I_{\mathsf{A'B'}}}\mathsf{SWAP}_{\mathsf{BB'}} \cdot \rbra*{\ketbra{\rho}{\rho}_{\mathsf{AB}}\otimes \ketbra{\sigma}{\sigma}_{\mathsf{A'B'}}}} \\
    ={}& \tr\rbra[\big]{\mathsf{SWAP}_{\mathsf{AA'}}\rbra*{I_{\mathsf{AB}}\otimes \ketbra{\sigma}{\sigma}_{\mathsf{A'B'}}}\mathsf{SWAP}_{\mathsf{AA'}} \cdot \rbra*{\ketbra{\rho}{\rho}_{\mathsf{AB}}\otimes \ketbra{\sigma}{\sigma}_{\mathsf{A'B'}}}} \\
    ={}& \tr\rbra[\Big]{\tr_{\mathsf{B}}\rbra[\big]{\mathsf{SWAP}_{\mathsf{AA'}}\rbra*{I_{\mathsf{AB}}\otimes \ketbra{\sigma}{\sigma}_{\mathsf{A'B'}}}\mathsf{SWAP}_{\mathsf{AA'}} \cdot \rbra*{\ketbra{\rho}{\rho}_{\mathsf{AB}}\otimes \ketbra{\sigma}{\sigma}_{\mathsf{A'B'}}}}} \\
    ={}& \tr\rbra[\big]{\mathsf{SWAP}_{\mathsf{AA'}}\rbra*{I_{\mathsf{A}}\otimes \ketbra{\sigma}{\sigma}_{\mathsf{A'B'}}}\mathsf{SWAP}_{\mathsf{AA'}} \cdot \rbra*{\tr_{\mathsf{B}}\rbra*{\ketbra{\rho}{\rho}_{\mathsf{AB}}}\otimes \ketbra{\sigma}{\sigma}_{\mathsf{A'B'}}}} \\
    ={}& \tr\rbra[\big]{\mathsf{SWAP}_{\mathsf{AA'}}\rbra*{I_{\mathsf{A}}\otimes \ketbra{\sigma}{\sigma}_{\mathsf{A'B'}}}\mathsf{SWAP}_{\mathsf{AA'}}\cdot\rbra*{\rho_{\mathsf{A}}\otimes \ketbra{\sigma}{\sigma}_{\mathsf{A'B'}}}} \\
    ={}& \tr\rbra[\big]{\rbra*{I_{\mathsf{A}}\otimes\bra{\sigma}_{\mathsf{A'B'}}}\cdot\mathsf{SWAP}_{\mathsf{AA'}}\rbra*{I_{\mathsf{A}}\otimes \ketbra{\sigma}{\sigma}_{\mathsf{A'B'}}}\mathsf{SWAP}_{\mathsf{AA'}}\cdot\rbra*{\rho_{\mathsf{A}}\otimes \ket{\sigma}_{\mathsf{A'B'}}}} \\
    ={}& \tr\rbra*{\rbra[\big]{\rbra*{I_{\mathsf{A}}\otimes\bra{\sigma}_{\mathsf{A'B'}}}\mathsf{SWAP}_{\mathsf{AA'}}\rbra*{I_{\mathsf{A}}\otimes \ket{\sigma}_{\mathsf{A'B'}}}}^2 \cdot {\rho_{\mathsf{A}}}}.
\end{align*}
The proof is completed by showing that $\rbra*{I_{\mathsf{A}}\otimes\bra{\sigma}_{\mathsf{A'B'}}}\mathsf{SWAP}_{\mathsf{AA'}}\rbra*{I_{\mathsf{A}}\otimes \ket{\sigma}_{\mathsf{A'B'}}} = \sigma_{\mathsf{A}}$. 
To see this, note that
\begin{align*}
    & \rbra*{I_{\mathsf{A}}\otimes\bra{\sigma}_{\mathsf{A'B'}}}\mathsf{SWAP}_{\mathsf{AA'}}\rbra*{I_{\mathsf{A}}\otimes \ket{\sigma}_{\mathsf{A'B'}}} \\
    ={} & \tr_{\mathsf{A'B'}}\rbra[\big]{\mathsf{SWAP}_{\mathsf{AA'}} \cdot \rbra*{I_{\mathsf{A}}\otimes \ketbra{\sigma}{\sigma}_{\mathsf{A'B'}}}} \\
    ={} & \tr_{\mathsf{A'}}\rbra[\Big]{\mathsf{SWAP}_{\mathsf{AA'}} \cdot \rbra[\big]{I_{\mathsf{A}}\otimes \tr_{\mathsf{B'}}\rbra*{\ketbra{\sigma}{\sigma}_{\mathsf{A'B'}}}}} \\
    ={} & \tr_{\mathsf{A'}}\rbra[\big]{\mathsf{SWAP}_{\mathsf{AA'}}\cdot\rbra*{I_{\mathsf{A}}\otimes \sigma_{\mathsf{A'}}}} \\
    ={} & \tr_{\mathsf{A'}}\rbra[\big]{\rbra*{\sigma_{\mathsf{A}}\otimes I_{\mathsf{A'}}}\cdot\mathsf{SWAP}_{\mathsf{AA'}}} \\
    ={} & \sigma_{\mathsf{A}} \cdot \tr_{\mathsf{A'}}\rbra{\mathsf{SWAP}_{\mathsf{AA'}}} \\
    ={} & \sigma_{\mathsf{A}} \cdot I_{\mathsf{A}} \\
    ={} & \sigma_{\mathsf{A}}.
\end{align*}
\end{proof}

Therefore, with a similar analysis, we can estimate the quantity $\sqrt{\tr\rbra{\rho\sigma^2}}$ by \cref{prop:prob-mixed} with the same complexity as \cref{thm:fidelity_to_pure}, stated as follows.

\begin{theorem}\label{thm:trace_rho_sigma2}
    Suppose that $U$ and $V$ are quantum unitary operators that prepare $n$-qubit purifications of $k$-qubit mixed states $\rho = \tr_\mathsf{B}\rbra{\ketbra{\rho}{\rho}_{\mathsf{AB}}}$ and $\sigma =\tr_\mathsf{B'}\rbra{\ketbra{\sigma}{\sigma}_{\mathsf{A'B'}}}$, respectively, as in \cref{eq:purification_rho,eq:purification_sigma}.
    For $\varepsilon \in \rbra{0,1}$, there is a quantum query algorithm that estimates $\sqrt{\tr\rbra{\rho\sigma^2}}$ to within additive error $\varepsilon$ with probability at least $2/3$ using $O\rbra{1/\varepsilon}$ queries to \mbox{(controlled-)$U$}, \mbox{(controlled-)$V$}, and their inverses and $O\rbra{n/\varepsilon}$ elementary quantum gates in addition to oracle queries.
\end{theorem}

To provide an intuitive understanding of our encoding scheme, we restructure the circuit in \cref{fig:q-von-ex} into the form shown in \cref{fig:q-von-ex2}.
\begin{figure} [htp]
    \centering
    \begin{quantikz}
        \lstick{$\ket{0}_{\mathsf{A}}$} & \gate[2]{U} & \gategroup[wires=4,steps=5,style={dashed}]{$\sigma_{\mathsf{A}}\otimes I_{\mathsf{B}}$} & & \permute{3,2,1,4} & & & &\\
        \lstick{$\ket{0}_{\mathsf{B}}$} & & & & & & & & \\
        \setwiretype{n} & & \midstick{$\ket{0}_{\mathsf{A'}}$} & \gate[2]{V} \setwiretype{q} & & \gate[2]{V^{\dagger}} & \midstick{$\ket{0}_{\mathsf{A'}}$} \\
        \setwiretype{n} & & \midstick{$\ket{0}_{\mathsf{B'}}$} & \setwiretype{q} & & & \midstick{$\ket{0}_{\mathsf{B'}}$}
    \end{quantikz}
    \caption{Restructured circuit of \cref{fig:q-von-ex}, illustrating the operator $\sigma_{\mathsf{A}}\otimes I_{\mathsf{B}}$ and its action on the purification $\ket{\rho}_{\mathsf{AB}}$.}
    \label{fig:q-von-ex2}
\end{figure}
This revised visualization explicitly reveals the construction of an operator $\sigma_{\mathsf{A}}\otimes I_{\mathsf{B}}$ acting on the purification $\ket{\rho}_{\mathsf{AB}}$ of $\rho_{\mathsf{A}}$.
Since
\[
\Abs*{\rbra*{\sigma_{\mathsf{A}}\otimes I_{\mathsf{B}}}\ket{\rho}_{\mathsf{AB}}}^2 = \tr\rbra*{\rho\sigma^2},
\]
\iffalse
\begin{align*}
    & \Abs*{\rbra*{\sigma_{\mathsf{A}}\otimes I_{\mathsf{B}}}\ket{\rho}_{\mathsf{AB}}}^2 \\
    ={}& \tr\rbra*{\rbra*{\sigma_{\mathsf{A}}\otimes I_{\mathsf{B}}}\ketbra{\rho}{\rho}_{\mathsf{AB}}\rbra*{\sigma_{\mathsf{A}}\otimes I_{\mathsf{B}}}} \\
    ={}& \tr\rbra*{\ketbra{\rho}{\rho}_{\mathsf{AB}}\rbra*{\sigma_{\mathsf{A}}^2\otimes I_{\mathsf{B}}}} \\
    ={}& \tr\rbra*{\tr_{\mathsf{B}}\rbra*{\ketbra{\rho}{\rho}_{\mathsf{AB}}\rbra*{\sigma_{\mathsf{A}}^2\otimes I_{\mathsf{B}}}}} \\
    ={}& \tr\rbra*{\tr_{\mathsf{B}}\rbra*{\ketbra{\rho}{\rho}_{\mathsf{AB}}}\sigma_{\mathsf{A}}^2} \\
    ={}& \tr\rbra*{\rho_{\mathsf{A}}\sigma_{\mathsf{A}}^2},
\end{align*}
\fi
we clearly see that $\sqrt{\tr\rbra{\rho\sigma^2}}$ is encoded into the amplitude of an efficiently preparable quantum state.
While the encoding pattern for $\sigma$ in \cref{fig:q-von-ex2} had previously appeared in~\cite{LC19}, people mainly focused on operator properties without explicitly deriving the amplitude structure of the resulting pure state.

\subsection{An implication: Generalized pure-state fidelity estimation} \label{sec:imp}

Our quantum algorithm provided in \cref{thm:main} suggests a new approach to estimating the fidelity between pure states, where the pure states are given through purified quantum query access, or equivalently prepared by (purified) quantum channels. 
The purified quantum query access to a (mixed) quantum state can be understood as a purified version of the quantum channel that prepares the quantum state (cf.\ \cite{GLM+22}). 
Formally, a quantum channel $\mathcal{E}$ is said to prepare a quantum state $\rho$, if $\rho = \mathcal{E}\rbra{\ketbra{0}{0}}$.
A purified version of $\mathcal{E}$ is a unitary operator $U$ acting on two subsystems $\mathsf{A}$ and $\mathsf{B}$ such that $\mathcal{E}\rbra{\sigma} = \tr_{\mathsf{B}} \rbra{ U \rbra{\sigma_{\mathsf{A}} \otimes \ketbra{0}{0}_{\mathsf{B}}} U^\dag }$ for every mixed quantum state $\sigma$. 
Then, $\rho$ can be prepared by the unitary operator $U$ in the way that its purification is prepared by $\ket{\rho}_{\mathsf{AB}} = U \ket{0}_{\mathsf{A}} \ket{0}_{\mathsf{B}}$ and $\rho$ is obtained by tracing out the subsystem $\mathsf{B}$, i.e., $\rho = \tr_{\mathsf{B}}\rbra{\ketbra{\rho}{\rho}_{\mathsf{AB}}}$. 
In the following, we state the special case of \cref{thm:main} for pure-state fidelity estimation.

\begin{corollary} [Pure-state fidelity estimation]\label{cor:pure-fidelity}
    Given purified quantum query access to two $n$-qubit pure states $\ket{\varphi}$ and $\ket{\psi}$, the fidelity $\mathrm{F}\rbra{\ket{\varphi}, \ket{\psi}} = \abs{\braket{\varphi}{\psi}}$ can be estimated to within additive error $\varepsilon$ with query complexity $O\rbra{1/\varepsilon}$ and $O\rbra{n/\varepsilon}$ elementary quantum gates in addition to oracle queries.
\end{corollary}

Moreover, our approach also applies to estimating the trace distance between two pure states.
This is obtained by noting that if $\rho = \ketbra{\varphi}{\varphi}$ is also a pure state, the circuit in \cref{fig:q-von-ex} outputs tow bits $x_{\mathsf{A}}, x_{\mathsf{B}} \in \cbra{0, 1}$ such that 
\[
\Pr\sbra{x_{\mathsf{A}} + x_{\mathsf{B}} \neq 0} = 1 - \mathrm{F}^2\rbra{\ket{\varphi}, \ket{\psi}} = \rbra*{ \frac{1}{2}\Abs*{ \ketbra{\varphi}{\varphi} - \ketbra{\psi}{\psi} }_1 }^2,
\]
which is exactly the square of the trace distance between $\ket{\varphi}$ and $\ket{\psi}$.
Using an argument similar to the proof of \cref{thm:fidelity_to_pure}, we have the following quantum algorithm for pure-state trace distance estimation.

\begin{corollary}[Pure-state trace distance estimation] \label{cor:pure-td}
    Given purified quantum query access to two $n$-qubit pure states $\ket{\varphi}$ and $\ket{\psi}$, the trace distance $\frac{1}{2} \Abs{ \ketbra{\varphi}{\varphi} - \ketbra{\psi}{\psi} }_1 = \sqrt{1 - \abs{\braket{\varphi}{\psi}}^2}$ can be estimated to within additive error $\varepsilon$ with query complexity $O\rbra{1/\varepsilon}$ and $O\rbra{n/\varepsilon}$ elementary quantum gates in addition to oracle queries.
\end{corollary}

\cref{cor:pure-fidelity,cor:pure-td} respectively generalize the pure-state fidelity estimation in \cite[Theorem IV.2]{Wan24} and the pure-state trace distance estimation in \cite[Theorem IV.1]{Wan24} to a more general input model.
By comparison, the prior approach for pure-state fidelity and trace distance estimation \cite{Wan24} requires that the pure states should be prepared by a unitary quantum channel.
Specifically, Ref.~\cite{Wan24} assumes two unitary oracles $U_{\varphi}$ and $U_{\psi}$ such that $U_{\varphi} \ket{0} = \ket{\varphi}$ and $U_{\psi} \ket{0} = \ket{\psi}$ (see \cref{sec:pure-fidelity} for details). 
The input model employed in \cite{Wan24} is more restricted than the purified quantum query access model.
For example, $U_{\varphi} \otimes I$ always provides purified quantum query access to $\ket{\varphi}$, where $I$ is the identity operator; however, it is unlikely that $U_{\varphi}$ could be implemented by purified quantum query access to $\ket{\varphi}$.

In addition to adopting a more general input model than \cite{Wan24}, \cref{cor:pure-fidelity,cor:pure-td} attain the same optimal query complexity of $O(1/\varepsilon)$, where the lower bound (noted in \cite{Wan24}) is implied in \cite{BBC+01,NW99} (see \cref{thm:lower-bound-pure}).

\section{Sample-Optimal Approach} \label{sec:sample}

As an important implication of our query-optimal approach to estimating the fidelity to a pure state in \cref{sec:algo}, we can further obtain a sample-optimal approach to the same task by quantum sample-to-query lifting \cite{WZ25a,WZ25b,TWZ25,CWZ25}. 

\begin{theorem} [Estimating fidelity to a pure state given sample access] \label{thm:est-fid-sampl}
    Suppose that $\rho$ is an unknown $k$-qubit mixed quantum state and $\ket{\psi}$ is an unknown $k$-qubit pure state. 
    For $\varepsilon \in \rbra{0, 1}$, there is a quantum sample algorithm that estimates $\mathrm{F}\rbra{\rho, \ket{\psi}}$ to within additive error $\varepsilon$ with probability at least $2/3$ using $O\rbra{1/\varepsilon^2}$ samples of $\rho$ and $\ket{\psi}$ and $\widetilde{O}\rbra{\poly\rbra{k}/\varepsilon^8}$ elementary quantum gates. 
\end{theorem}

\subsection{Estimating fidelity to a pure state} \label{sec:sample_pure}

To prove \cref{thm:est-fid-sampl}, we use the following version of quantum sample-to-query lifting adapted from \cite{UNWT25}, which refines the approach in \cite{TWZ25} consisting of three main steps: state conditionalization with phase \cite{Kre21,GZ25} (equipped with density matrix exponentiation \cite{LMR14,KLL+17,GKP+25}), unitary controllization \cite{TW25}, and random purification (based on the results in \cite{SW22,CWZ24}).
In particular, the last step dominates the number of (additional) elementary gates due to the use of the Schur transform \cite{Har05,BCH06,BCH07,Kro19,BFG+25} with the current best implementation in \cite{BFG+25}. 

For readability, we introduce the following norms. 
For a (finite-dimensional) matrix $A$, the trace norm (also known as the Schatten $1$-norm) of $A$ is defined by $\Abs{A}_1 = \tr\rbra{\sqrt{A^\dag A}}$. 
For a linear operator $\mathcal{E}$ that maps a matrix $A$ to another matrix $\mathcal{E}\rbra{A}$ of the same dimension (where $\mathcal{E}$ is also called a superoperator), the trace norm (also known as the induced trace norm in \cite{Wat18}) is defined by $\Abs{\mathcal{E}}_{\tr} = \sup_{A \colon \Abs{A}_1 \leq 1} \Abs{\mathcal{E}\rbra{A}}_1$. 

\begin{theorem}[Quantum sample-to-query lifting, adapted from {\cite[Theorem 1.5]{TWZ25}} and {\cite[Theorem 26]{UNWT25}}] \label{thm:query-to-sample}
    Given purified quantum query access to a $\rbra{2k+1}$-qubit unitary oracle $U$ that prepares the purification of an unknown $k$-qubit mixed quantum state $\rho$, for any quantum query algorithm $A^U$ that uses $Q$ queries to \mbox{(controlled-)$U$} and \mbox{(controlled-)$U^\dag$}, $A^U$ can be approximated to precision $0.01$ (in the trace norm distance) by a quantum channel $\mathcal{E}$ implemented by $O\rbra{Q^2}$ samples of $\rho$ and $\widetilde{O}\rbra{Q^8 \poly\rbra{k}}$ additional elementary quantum gates. 
    
    More specifically, if $A^U$ has the form
    \[
    A^U = G_Q \cdot F_Q \cdot \dots \cdot G_1 \cdot F_1 \cdot G_0,
    \]
    where $G_0, \dots, G_Q$ are unitary operators independent of $U$ and $F_1, \dots, F_Q$ are either \mbox{(controlled-)$U$} or \mbox{(controlled-)$U^\dag$} acting on certain qubits, then the quantum channel $\mathcal{E}$ satisfies that
    \[
    \Abs*{ \mathcal{A}^U - \mathcal{E}}_{\textup{tr}} \leq 0.01,
    \]
    where $\mathcal{A}^U \colon \sigma \mapsto A^U \sigma \rbra{A^U}^\dag$ is the quantum unitary channel induced by $A^U$, 
    and $\mathcal{E}$ is of the form
    \[
    \mathcal{E} = \mathcal{G}_Q \circ \mathcal{F}_Q \circ \dots \circ \mathcal{G}_1 \circ \mathcal{F}_1 \circ \mathcal{G}_0,
    \]
    where $\circ$ denotes the composition of superoperators, $\mathcal{G}_j \colon \sigma \mapsto G_j \sigma G_j^\dag$ is the quantum unitary channel induced by the unitary operator $G_j$ for each $0 \leq j \leq Q$, and $\mathcal{F}_1, \dots, \mathcal{F}_Q$ are implemented by $O\rbra{Q^2}$ samples of $\rho$ and $\widetilde{O}\rbra{Q^8 \poly\rbra{k}}$ elementary quantum gates in total. 
\end{theorem}

The number of qubits of the unitary oracle $U$ required in \cref{thm:query-to-sample} is compatible with the query algorithm in \cref{thm:fidelity_to_pure}.
Indeed, every $k$-qubit quantum state $\rho$ can be purified using $k$ additional qubits (cf.\ \cite[Section 2.5]{NC10}), and appending one ancilla qubit gives a $\rbra{2k+1}$-qubit purification of the $k$-qubit quantum state $\rho$.
For the pure state $\ket{\psi}$, we may use the purification $\ket{\psi}_{\mathsf{A}'}\ket{0}_{\mathsf{B}'}$ with $\mathsf{B}'$ consisting of $\rbra{k+1}$ qubits.
Thus, although \cref{thm:fidelity_to_pure} allows arbitrary $n$-qubit purifications, in the proof of \cref{thm:est-fid-sampl} below, we instantiate it with $n=2k+1$, so both $U$ and $V$ satisfy the dimensional assumption of \cref{thm:query-to-sample}.

Now we are ready to prove \cref{thm:est-fid-sampl}.

\begin{proof}[Proof of \cref{thm:est-fid-sampl}]
    Suppose that $U$ and $V$ are $\rbra{2k+1}$-qubit unitary oracles that prepare purifications of $\rho$ and $\ketbra{\psi}{\psi}$, respectively.
    The final quantum algorithm constructed below will use only samples of $\rho$ and $\ketbra{\psi}{\psi}$. 
    These oracles $U$ and $V$ are only used to witness how the quantum query algorithm is converted to the final quantum algorithm (that uses only samples).
    By \cref{thm:fidelity_to_pure}, for any $\varepsilon\in \rbra{0,1}$, there is a quantum query algorithm $A^{U,V}$ that estimates $\mathrm{F}\rbra{\rho, \ket{\psi}}$ to within additive error $\varepsilon$ with probability at least $0.99$ using $O\rbra{1/\varepsilon}$ queries to \mbox{(controlled-)$U$}, \mbox{(controlled-)$V$}, and their inverses.

    We first eliminate the queries to $U$ while treating all queries to $V$ as part of the quantum gates that are independent of $U$. 
    To make this explicit, $A^{U,V}$ has the form
    \begin{equation*}
        A^{U,V} = G_{Q_U}^{V}\cdot F_{Q_U} \cdot \cdots \cdot G_1^{V}\cdot F_{1} \cdot G_0^V,
    \end{equation*}
    where $Q_U = O\rbra{1/\varepsilon}$, $F_1, \dots, F_{Q_U}$ are either \mbox{(controlled-)$U$} or \mbox{(controlled-)$U^\dag$} acting on certain qubits, and $G_0^V,\ldots, G_{Q_U}^{V}$ are unitary operators independent of $U$.
    Across all the unitary operators $G_0^V,\ldots,G_{Q_U}^V$, the total number of queries to \mbox{(controlled-)$V$} and \mbox{(controlled-)$V^\dag$} is $Q_V=O\rbra{1/\varepsilon}$ and the total number of elementary quantum gates is $O\rbra{k/\varepsilon}$, because here $n=2k+1$.
    By \cref{thm:query-to-sample}, there is a quantum channel $\mathcal{E}^V$ of the form
    \begin{equation*}
        \mathcal{E}^V = \mathcal{G}_{Q_U}^{V}\circ \mathcal{F}_{Q_U} \circ \cdots \circ \mathcal{G}_1^{V}\circ \mathcal{F}_1 \circ \mathcal{G}_0^V
    \end{equation*}
    such that
    \begin{equation}\label{eq:U_replaced1}
        \Abs*{\mathcal{A}^{U,V} - \mathcal{E}^V}_{\textup{tr}} \leq 0.01,
    \end{equation}
    where $\mathcal{A}^{U,V} \colon \sigma \mapsto A^{U,V} \sigma \rbra{A^{U,V}}^\dag$ is the quantum unitary channel induced by $A^{U,V}$, $\mathcal{G}_j^V$ is the quantum unitary channel induced by the unitary operator $G_j^V$ for each $0 \leq j \leq Q_U$ and $\mathcal{F}_1, \dots, \mathcal{F}_{Q_U}$ are implemented by
    $
        O\rbra{Q_U^2} = O\rbra{1/\varepsilon^2}
    $
    samples of $\rho$ and $\widetilde{O}\rbra{Q_U^8 \poly\rbra{k}} = \widetilde{O}\rbra{\poly\rbra{k}/\varepsilon^8}$ elementary quantum gates in total.

    Let $s_j$ be the number of samples of $\rho$ consumed by $\mathcal{F}_j$ for $1\leq j\leq Q_U$ and $s = s_1+\cdots s_{Q_U} = O\rbra{1/\varepsilon^2}$, then for each $\mathcal{F}_j$, there is a unitary $\hat{F}_j$ such that $\mathcal{F}_j$ is implemented as
    \begin{equation*}
        \mathcal{F}_j\rbra{\sigma} = \tr_{\textsf{anc}}\rbra*{\hat{F}_j\rbra*{\rbra*{\rho^{\otimes s_j}\otimes \ketbra{0}{0}^{\otimes t_j}}_{\textsf{anc}}\otimes \sigma}\hat{F}_j^{\dagger}}
    \end{equation*}
    for any mixed quantum state $\sigma$,
    with the ancilla register initialized in the state $\rho^{\otimes s_j}\otimes \ketbra{0}{0}^{\otimes t_j}$.
    Then, letting $t = t_1+\cdots+t_{Q_U}$, we can construct a unitary operator $B^V$ that uses the same $Q_V = O\rbra{1/\varepsilon}$ queries to \mbox{(controlled-)$V$} and \mbox{(controlled-)$V^\dag$} as $\mathcal{E}^V$ such that
    \begin{equation}\label{eq:U_replaced2}
        \mathcal{E}^V\rbra{\sigma} = \tr_{\textsf{anc}}\rbra*{B^V\rbra*{\rbra*{\rho^{\otimes s}\otimes \ketbra{0}{0}^{\otimes t}}_{\textsf{anc}}\otimes \sigma}\rbra*{B^V}^{\dagger}}
    \end{equation}
    for any mixed quantum state $\sigma$.
    Now the only remaining oracle dependence in $B^V$ is on the single oracle $V$, where the registers containing $\rho^{\otimes s}$ are independent of $V$.
    Therefore, applying \cref{thm:query-to-sample} a second time, now to $B^V$ as a quantum query algorithm with query access to $V$ for the $k$-qubit state $\ketbra{\psi}{\psi}$, yields a quantum channel $\mathcal{K}$ implemented by $O\rbra{Q_V^2} = O\rbra{1/\varepsilon^2}$ samples of $\ketbra{\psi}{\psi}$ and $\widetilde{O}\rbra{Q_V^8 \poly\rbra{k}} = \widetilde{O}\rbra{\poly\rbra{k}/\varepsilon^8}$ additional elementary quantum gates,
    satisfying
    \begin{equation}\label{eq:UV_replaced}
        \Abs*{\mathcal{B}^{V} - \mathcal{K}}_{\textup{tr}} \leq 0.01,
    \end{equation}
    where $\mathcal{B}^V \colon \sigma \mapsto B^V \sigma \rbra{B^V}^\dag$ is the quantum unitary channel induced by $B^V$. 

    Summarizing all the results above, we define a quantum channel
    \begin{equation*}
        \mathcal{J}(\sigma) = \tr_{\textsf{anc}}\rbra*{\mathcal{K}\rbra*{\rbra*{\rho^{\otimes s}\otimes \ketbra{0}{0}^{\otimes t}}_{\textsf{anc}}\otimes \sigma}}
    \end{equation*}
    for any input state $\sigma$.
    Since $s = O\rbra{1/\varepsilon^2}$, $\mathcal{J}$ can be implemented using $O\rbra{Q_U^2} + O\rbra{Q_V^2} = O\rbra{1/\varepsilon^2}$ samples of $\rho$ and $\ketbra{\psi}{\psi}$ and $\widetilde{O}\rbra{Q_U^8 \poly\rbra{k}} + \widetilde{O}\rbra{Q_V^8 \poly\rbra{k}} + O\rbra{k/\varepsilon} = \widetilde{O}\rbra{\poly\rbra{k}/\varepsilon^8}$ elementary quantum gates.
    By invoking the contractivity of trace norm distance and \cref{eq:U_replaced1,eq:U_replaced2,eq:UV_replaced}, we have
    \begin{equation*}
        \Abs*{\mathcal{A}^{U,V} - \mathcal{J}}_{\tr} \leq 0.02.
    \end{equation*}
    We assume that $A^{U, V}$ acts on two registers $\mathsf{O}$ (for output) and $\mathsf{W}$ (for work). 
    To retrieve information from the output state $A^{U, V}\ket{\bar 0}_{\mathsf{OW}}$, we denote the computational basis measurement on $\mathsf{O}$ by $\cbra{\Pi_x = \ketbra{x}{x}_{\mathsf{O}} \otimes I_{\mathsf{W}}}$, where $x$ ranges over all possible (finitely many) outcomes. 
    Then, the probability that $A^{U, V}$ (on input $\ket{\bar 0}_{\mathsf{OW}}$) outputs $x$ is given by
    \[
    p_x = \tr\rbra*{\Pi_x \mathcal{A}^{U,V}\rbra{\ketbra{\bar0}{\bar0}_{\mathsf{OW}}}} = \Abs*{\Pi_x A^{U, V}\ket{\bar 0}_{\mathsf{OW}}}^2.
    \]
    Let $\tilde p_x$ be the probability that $\mathcal{J}$ (on input $\ket{\bar 0}_{\mathsf{OW}}$) outputs $x$, i.e., 
    \[
    \tilde p_x = \tr\rbra*{ \Pi_x \mathcal{J}\rbra{\ketbra{\bar 0}{\bar 0}_{\mathsf{OW}}} }.
    \]
    Then, the total variation distance between the two probability distributions $p$ and $\tilde p$ can be bounded by
    \[
    d_{\textup{TV}}\rbra{p, \tilde p} = \frac{1}{2} \sum_{x} \abs*{p_x - \tilde p_x} \leq \frac{1}{2} \Abs*{\mathcal{A}^{U,V} - \mathcal{J}}_{\tr} \leq 0.01.
    \]
    According to the condition that $A^{U,V}$ can estimate $\mathrm{F}\rbra{\rho, \ket{\psi}}$ to within additive error $\varepsilon$ with probability at least $0.99$, we have
    \[
    \sum_{x \colon \abs{x - \mathrm{F}\rbra{\rho, \ket{\psi}}} \leq \varepsilon} p_x \geq 0.99. 
    \]
    Therefore, $\mathcal{J}$ can estimate $\mathrm{F}\rbra{\rho, \ket{\psi}}$ to within additive error $\varepsilon$ with probability
    \begin{align*}
        \sum_{x \colon \abs{x - \mathrm{F}\rbra{\rho, \ket{\psi}}} \leq \varepsilon} \tilde p_x 
        & \geq \sum_{x \colon \abs{x - \mathrm{F}\rbra{\rho, \ket{\psi}}} \leq \varepsilon} \rbra*{p_x - \abs*{p_x - \tilde p_x}} \\
        & \geq \rbra*{\sum_{x \colon \abs{x - \mathrm{F}\rbra{\rho, \ket{\psi}}} \leq \varepsilon} p_x} - 2 d_{\textup{TV}}\rbra{p, \tilde p} \\
        & \geq 0.99 - 2 \cdot 0.01 \\
        & = 0.97.
    \end{align*}
\end{proof}

\begin{remark}[The necessity of the ``purified quantum query access'' required in \cref{thm:main}] \label{remark:sample}
    Technically, the proof of \cref{thm:est-fid-sampl} is by converting the query complexity $O\rbra{1/\varepsilon}$ in \cref{thm:fidelity_to_pure} to a sample complexity of $O\rbra{1/\varepsilon^2}$ by quantum sample-to-query lifting (see \cref{thm:query-to-sample}). 
    One may wonder if the sample complexity $O\rbra{1/\varepsilon^2}$ for estimating the fidelity $\mathrm{F}\rbra{\ket{\varphi}, \ket{\psi}}$ between two pure states can be obtained in the same way using the quantum query algorithm for pure-state fidelity estimation with query complexity $O\rbra{1/\varepsilon}$ given in \cite{Wan24}. 
    Here, we clarify that the result in \cite{Wan24}, however, is not sufficient to establish the sample complexity for pure-state fidelity estimation. 
    This is because the query complexity given in \cref{thm:fidelity_to_pure} is under the purified quantum query access model (see \cref{def:query}) to which the quantum sample-to-query lifting (\cref{thm:query-to-sample}) can apply, whereas the query complexity result given in \cite{Wan24} assumes quantum query oracles of the form $U\ket{0} = \ket{\psi}$ which are not under the purified quantum query access model. 
\end{remark}

\subsection{Generalizations} \label{sec:general-sample}

Using the same argument as for the proof of \cref{thm:est-fid-sampl} with the result in \cref{thm:trace_rho_sigma2}, we have the following generalization of \cref{thm:est-fid-sampl}. 

\begin{theorem} \label{thm:trace_rho_sigma2-sample}
    Suppose that $\rho$ and $\sigma$ are two unknown $k$-qubit mixed quantum states. 
    For $\varepsilon \in \rbra{0, 1}$, there is a quantum sample algorithm that estimates $\sqrt{\tr\rbra{\rho\sigma^2}}$ to within additive error $\varepsilon$ with probability at least $2/3$ using $O\rbra{1/\varepsilon^2}$ samples of $\rho$ and $\sigma$ and $\widetilde{O}\rbra{\poly\rbra{k}/\varepsilon^8}$ elementary quantum gates. 
\end{theorem}

\begin{proof}
    The proof is based on \cref{thm:trace_rho_sigma2} and uses the same argument as the proof of \cref{thm:est-fid-sampl} (which uses \cref{thm:fidelity_to_pure}). 
\end{proof}

\section{Lower Bounds} \label{sec:lb}

For completeness, we discuss the complexity lower bounds for estimating the fidelity. 
In fact, as noted in \cite{Wan24}, when both quantum states are pure, (i) given query access, a query complexity lower bound of $\Omega\rbra{1/\varepsilon}$ is implied in \cite{BBC+01,NW99}, and (ii) given sample access, a sample complexity lower bound of $\Omega\rbra{1/\varepsilon^2}$ was given in \cite{ALL22}. 

\begin{theorem} [Adapted from \cite{BBC+01,NW99,ALL22}]\label{thm:lower-bound-pure} We have the following lower bounds for fidelity estimation. 
\begin{enumerate}[label=(\arabic*)]
    \item \textbf{Query complexity} \cite{BBC+01,NW99}: Given purified quantum query access to quantum states $\rho$ and $\ket{\psi}$, any quantum query algorithm that estimates $\mathrm{F}\rbra{\rho, \ket{\psi}}$ to within additive error $\varepsilon$ requires query complexity $\Omega\rbra{1/\varepsilon}$, even if $\rho$ is pure. \label{item:pure-query}
    \item \textbf{Sample complexity} \cite{ALL22}: Given sample access to quantum states $\rho$ and $\ket{\psi}$, any quantum sample algorithm that estimates $\mathrm{F}\rbra{\rho, \ket{\psi}}$ to within additive error $\varepsilon$ requires sample complexity $\Omega\rbra{1/\varepsilon^2}$, even if $\rho$ is pure. \label{item:pure-sample}
\end{enumerate}
\end{theorem}

We strengthen this lower bound so that it applies to the case where $\rho$ is a mixed quantum state of an arbitrary rank. 

\begin{theorem} \label{thm:lower-bound-rank}
We have the following lower bounds for fidelity estimation. 
\begin{enumerate}[label=(\arabic*)]
    \item \textbf{Query complexity}: Given purified quantum query access to quantum states $\rho$ and $\ket{\psi}$, any quantum query algorithm that estimates $\mathrm{F}\rbra{\rho, \ket{\psi}}$ to within additive error $\varepsilon$ requires query complexity $\Omega\rbra{1/\varepsilon}$, even if $\rho$ is of an arbitrary rank. \label{item:query}
    \item \textbf{Sample complexity}: Given sample access to quantum states $\rho$ and $\ket{\psi}$, any quantum sample algorithm that estimates $\mathrm{F}\rbra{\rho, \ket{\psi}}$ to within additive error $\varepsilon$ requires sample complexity $\Omega\rbra{1/\varepsilon^2}$, even if $\rho$ is of an arbitrary rank. \label{item:sample}
\end{enumerate}
\end{theorem}

To show \cref{thm:lower-bound-rank}, we first prove the sample complexity lower bound by the Helstrom-Holevo bound \cite{Hel67,Hol73}, and then prove the query complexity lower bound by quantum sample-to-query lifting \cite{WZ25a,WZ25b,TWZ25}.

\subsection{Sample complexity lower bounds}

To show \cref{thm:lower-bound-rank}\ref{item:sample}, we need the Helstrom-Holevo bound \cite{Hel67,Hol73} as follows. 
Here, we use the version in \cite{Wil13}.

\begin{theorem}[Helstrom-Holevo bound, cf.\ {\cite[Section 9.1.4]{Wil13}}] \label{thm:hh}
    Let $\varrho$ be a mixed quantum state that is either $\rho_0$ or $\rho_1$ with equal probability. 
    By measuring $\varrho$ with the quantum measurement $\Lambda = \cbra{\Lambda_0, \Lambda_1}$ with $\Lambda_0, \Lambda_1$ positive semi-definite and $\Lambda_0 + \Lambda_1 = I$, the success probability of distinguishing whether $\varrho = \rho_0$ or $\varrho = \rho_1$ is bounded by
    \[
    p_{\textup{succ}} = \frac{1}{2}\tr\rbra{\Lambda_0\rho_0} + \frac{1}{2}\tr\rbra{\Lambda_1\rho_1} \leq \frac{1}{2} + \frac{1}{4} \Abs{\rho_0 - \rho_1}_1,
    \]
    where $\Abs{A}_1 = \tr\rbra{\sqrt{A^\dag A}}$ and $\frac{1}{2}\Abs{\rho_0-\rho_1}_1$ is the trace distance between $\rho_0$ and $\rho_1$.
\end{theorem}

We also need the Fuchs-van de Graaf inequality \cite{FvdG99} for trace distance and fidelity. 

\begin{theorem}[Fuchs-van de Graaf inequality, {\cite[Theorem 1]{FvdG99}}] \label{thm:FvdG99}
    For any two mixed quantum states $\rho$ and $\sigma$ of the same dimension, 
    \[
    1 - \mathrm{F} \rbra{\rho, \sigma} \leq \frac{1}{2}\Abs{\rho-\sigma}_1 \leq \sqrt{1 - \mathrm{F}^2 \rbra{\rho, \sigma}}.
    \]
\end{theorem}

Now we are ready to prove the sample complexity lower bound for fidelity estimation. 

\begin{proof}[Proof of \cref{thm:lower-bound-rank}\ref{item:sample}]
    Let $r = \rank\rbra{\rho} \geq 2$ be an arbitrary rank parameter. 
    We consider the problem of distinguishing the two quantum states 
    \begin{equation} \label{eq:def-rho-pm}
    \rho^{\pm} = \sum_{j=1}^{n} p_j^{\pm} \ketbra{j}{j},
    \end{equation}
    where 
    \begin{equation} \label{eq:def-prob}
    p^{\pm}_1 = p \pm \varepsilon, ~ p^{\pm}_j = \frac{1-p\mp\varepsilon}{r-1} \textup{ for } 2 \leq j \leq r, \textup{ and } p^{\pm}_{j} = 0 \textup{ for } r < j \leq n,
    \end{equation}
    where $p \in \rbra{0, 1}$ is an arbitrary constant. 
    It can be verified that
    \begin{equation} \label{eq:fidelity-pm}
        \mathrm{F}\rbra{\rho^+, \rho^-} = \sqrt{p^2 - \varepsilon^2} + \sqrt{\rbra{1 - p}^2 - \varepsilon^2} \geq 1 - \Theta\rbra{\varepsilon^2}.
    \end{equation}
    By \cref{thm:hh}, any quantum algorithm that determines whether an unknown quantum state $\varrho$ is $\rho^+$ or $\rho^-$ with success probability at least $2/3$ using $S$ samples of $\varrho$ requires 
    \[
    \frac{2}{3} \leq \frac{1}{2} + \frac{1}{4}\Abs*{\rbra{\rho^+}^{\otimes S} - \rbra{\rho^-}^{\otimes S}}_1,
    \]
    which gives (using \cref{thm:FvdG99,eq:fidelity-pm})
    \[
    S \geq \Omega\rbra*{\frac{1}{1-\mathrm{F}\rbra{\rho^+, \rho^-}}} \geq \Omega\rbra*{\frac{1}{\varepsilon^2}}.
    \]

    On the other hand, to determine whether an unknown quantum state $\varrho$ is $\rho^+$ or $\rho^-$, we can estimate the fidelity $\mathrm{F}\rbra{\varrho, \ket{\psi}}$ with $\ket{\psi} = \ket{1}$.
    Note that $\mathrm{F}\rbra{\rho^{\pm}, \ket{\psi}} = \sqrt{p\pm\varepsilon} = \sqrt{p} \pm \Theta\rbra{\varepsilon}$.
    Therefore, any quantum algorithm for estimating $\mathrm{F}\rbra{\varrho, \ket{\psi}}$ to within additive error $\Theta\rbra{\varepsilon}$ can be used to distinguish $\rho^+$ and $\rho^-$, thereby requiring sample complexity $\Omega\rbra{1/\varepsilon^2}$, even if $\varrho$ is of rank $r$. 
\end{proof}

\subsection{Query complexity lower bounds}

To show \cref{thm:lower-bound-rank}\ref{item:query}, we need the quantum sample-to-query lifting theorem \cite{WZ25a,WZ25b,TWZ25} as follows. 
Here, we use the version in \cite{CWZ25}. 

\begin{theorem}[Quantum sample-to-query lifting, cf.\ {\cite[Theorem 1.1]{CWZ25}}] \label{thm:lifting}
    Let $\mathcal{P}$ be a promise problem of quantum state testing. 
    Then, 
    \[
    \mathsf{Q}\rbra{\mathcal{P}} = \Omega\rbra*{\sqrt{\mathsf{S}\rbra{\mathcal{P}}}},
    \]
    where $\mathsf{Q}\rbra{\mathcal{P}}$ is the quantum query complexity of $\mathcal{P}$ in the purified quantum query access model and $\mathsf{S}\rbra{\mathcal{P}}$ is the sample complexity of $\mathcal{P}$. 
\end{theorem}

Now we are ready to prove the query complexity lower bound for fidelity estimation.

\begin{proof}[Proof of \cref{thm:lower-bound-rank}\ref{item:query}]
    Let $r = \rank\rbra{\rho} \geq 2$ be an arbitrary rank parameter. 
    In the proof of \cref{thm:lower-bound-rank}\ref{item:sample}, we have shown that $\mathsf{S}\rbra{\textsc{Dis}_{\rho^+, \rho^-}} = \Omega\rbra{1/\varepsilon^2}$, where $\textsc{Dis}_{\rho^+, \rho^-}$ is the problem of distinguishing the two quantum states $\rho^+$ and $\rho^-$ defined in \cref{eq:def-rho-pm}. 
    By \cref{thm:lifting}, we have $\mathsf{Q}\rbra{\textsc{Dis}_{\rho^+, \rho^-}} = \Omega\rbra{\sqrt{\mathsf{S}\rbra{\textsc{Dis}_{\rho^+, \rho^-}}}} = \Omega\rbra{1/\varepsilon}$. 

    On the other hand, using the same argument, to determine whether an unknown quantum state $\varrho$ is $\rho^+$ or $\rho^-$, we can estimate the fidelity $\mathrm{F}\rbra{\varrho, \ket{\psi}}$ with $\ket{\psi} = \ket{1}$.
    Note that $\mathrm{F}\rbra{\rho^{\pm}, \ket{\psi}} = \sqrt{p\pm\varepsilon} = \sqrt{p} \pm \Theta\rbra{\varepsilon}$.
    Therefore, any quantum algorithm for estimating $\mathrm{F}\rbra{\varrho, \ket{\psi}}$ to within additive error $\Theta\rbra{\varepsilon}$ can be used to distinguish $\rho^+$ and $\rho^-$, thereby requiring query complexity $\Omega\rbra{1/\varepsilon}$, even if $\varrho$ is of rank $r$. 
\end{proof}

\begin{remark}
    In an earlier version \cite{FW25c} of this paper, a different proof of \cref{thm:lower-bound-rank}\ref{item:query} was provided, which is achieved by reducing the problem of distinguishing probability distributions to fidelity estimation and using the quantum query complexity lower bound in \cite{Bel19}. 
    Specifically, the problem for the reduction is to distinguish the two probability distributions $p^{\pm}$ defined by \cref{eq:def-prob}, which requires quantum query complexity $\Omega\rbra{1/d_{\textup{H}}\rbra{p^+, p^-}} = \Omega\rbra{1/\varepsilon}$ by \cite[Theorem 4]{Bel19}, where $d_{\textup{H}}\rbra{p^+, p^-}$ is the Hellinger distance. 
    In comparison, our current proof provides a unified proof for both sample and query complexities by quantum sample-to-query lifting (\cref{thm:lifting}). 
    To see the similarity between the two proofs, note that $d_{\textup{H}}\rbra{p^+, p^-} = \sqrt{1 - \mathrm{F}\rbra{\rho^+, \rho^-}} = \Theta\rbra{\varepsilon}$. 
    On the other hand, both the two proofs can be understood from the same perspective, as the quantum query complexity lower bound for distinguishing probability distributions in \cite[Theorem 4]{Bel19} was shown to have another proof based on quantum sample-to-query lifting in \cite[Theorem 2.1]{CWZ25}. 
\end{remark}

\section{Discussion} \label{sec:discussion}

In this paper, we present two optimal quantum algorithms for estimating the fidelity of a mixed state to a pure state, where one is query-optimal and the other is sample-optimal. 
Our approaches are simple, which moreover estimates the quantity $\sqrt{\tr\rbra{\rho\sigma^2}}$ that has not been commonly considered in the literature. 
Here, we raise some questions for future research. 
\begin{enumerate}
    \item Given purified quantum query access, can we encode $\sqrt{\tr\rbra{\rho\sigma}}$ rather than $\sqrt{\tr\rbra{\rho\sigma^2}}$ in the amplitudes? 
    In addition to its relationship with fidelity estimation, this is related to the \textit{Frobenius norm} of a quantum state, $\Abs{\rho}_{F} = \sqrt{\tr\rbra{\rho^2}}$, which is the square root of purity. 
    \item Given sample access, can we estimate the fidelity $\mathrm{F}\rbra{\rho, \ket{\psi}}$ with time complexity better than the result $\widetilde{O}\rbra{\poly\rbra{k}/\varepsilon^8}$ given in \cref{thm:est-fid-sampl}?
    \item Can we find optimal distributed quantum algorithm for estimating the fidelity $\mathrm{F}\rbra{\ket{\varphi}, \ket{\psi}} = \abs{\braket{\varphi}{\psi}}$ between two pure states, generalizing the results in \cite{AS25,GHYZ24} which estimate the squared fidelity $\mathrm{F}^2\rbra{\ket{\varphi}, \ket{\psi}} = \abs{\braket{\varphi}{\psi}}^2$?
\end{enumerate}

\section*{Acknowledgment}

The work of Wang Fang was supported by the Engineering and Physical Sciences Research Council under Grant EP/X025551/1.
The work of Qisheng Wang was supported in part by the Startup Funding from Shanghai Jiao Tong University and in part by the Engineering and Physical Sciences Research Council under Grant EP/X026167/1.
Part of the work of Qisheng Wang was done when the author was with the School of Informatics, University of Edinburgh, Edinburgh, United Kingdom.

\addcontentsline{toc}{section}{References}

\bibliographystyle{plainurl}
\bibliography{main}

@inproceedings{Wat02,
    author = {Watrous, John},
    title = {Limits on the power of quantum statistical zero-knowledge},
    booktitle = {Proceedings of the 43rd Annual IEEE Symposium on Foundations of Computer Science},
    pages = {459--468},
    doi = {10.1109/SFCS.2002.1181970},
    year = {2002}
}

@inproceedings{GL20,
    author = {Gily{\'{e}}n, Andr{\'{a}}s and Li, Tongyang},
    title = {Distributional property testing in a quantum world},
    booktitle = {Proceedings of the 11th Innovations in Theoretical Computer Science Conference},
    pages = {25:1--25:19},
    doi = {10.4230/LIPIcs.ITCS.2020.25},
    year = {2020}
}

@article{GLM+22,
    author = {Gily{\'{e}}n, Andr{\'{a}}s and Lloyd, Seth and Marvian, Iman and Quek, Yihui and Wilde, Mark M.},
    title = {Quantum algorithm for {Petz} recovery channels and pretty good measurements},
    journal = {Physical Review Letters},
    volume = {128},
    number = {22},
    pages = {220502},
    doi = {10.1103/PhysRevLett.128.220502},
    year = {2022}
}

@article{BBC+01,
    author = {Beals, Robert and Buhrman, Harry and Cleve, Richard and Mosca, Michele and de Wolf, Ronald},
    title = {Quantum lower bounds by polynomials},
    journal = {Journal of the ACM},
    volume = {48},
    number = {4},
    pages = {778--797},
    doi = {10.1145/502090.502097},
    year = {2001}
}

@inproceedings{NW99,
    author = {Nayak, Ashwin and Wu, Felix},
    title = {The quantum query complexity of approximating the median and related statistics},
    booktitle = {Proceedings of the 31st Annual ACM Symposium on Theory of Computing},
    pages = {384--393},
    doi = {10.1145/301250.301349},
    year = {1999}
}

@article{Wan24,
    author = {Wang, Qisheng},
    title = {Optimal trace distance and fidelity estimations for pure quantum states},
    journal = {IEEE Transactions on Information Theory},
    volume = {70},
    number = {12},
    pages = {8791--8805},
    doi = {10.1109/TIT.2024.3447915},
    year = {2024}
}

@article{KLL+17,
    author = {Kimmel, Shelby and Lin, Cedric Yen-Yu and Low, Guang Hao and Ozols, Maris and Yoder, Theodore J.},
    title = {Hamiltonian simulation with optimal sample complexity},
    journal = {npj Quantum Information},
    volume = {3},
    number = {1},
    pages = {1--7},
    doi = {10.1038/s41534-017-0013-7},
    year = {2017}
}

@book{NC10,
    author = {Nielsen, Michael A. and Chuang, Isaac L.},
    title = {Quantum Computation and Quantum Information},
    publisher = {Cambridge University Press},
    year = {2010}
}

@incollection{MdW16,
    author = {Montanaro, Ashley and de Wolf, Ronald},
    title = {A survey of quantum property testing},
    year = {2016},
    publisher = {University of Chicago},
    booktitle = {Theory of Computing Library},
    series = {Graduate Surveys},
    number = {7},
    doi = {10.4086/toc.gs.2016.007},
    pages = {1--81},
}

@article{BCWdW01,
    author = {Buhrman, Harry and Cleve, Richard and Watrous, John and de Wolf, Ronald},
    title = {Quantum Fingerprinting},
    journal = {Physical Review Letters},
    volume = {87},
    number = {16},
    pages = {167902},
    doi = {10.1103/PhysRevLett.87.167902},
    year = {2001}
}

@incollection{BHMT02,
    author = {Brassard, Gilles and H{\o}yer, Peter and Mosca, Michele and Tapp, Alain},
    title = {Quantum amplitude amplification and estimation},
    editor = {Lomonaco, Jr., Samuel J. and Brandt, Howard E.},
    booktitle = {Quantum Computation and Information},
    volume = {305},
    number = {},
    pages = {53--74},
    doi = {10.1090/conm/305/05215},
    publisher = {AMS},
    series = {Contemporary Mathematics},
    year = {2002}
}

@article{FL11,
    author = {Flammia, Steven T. and Liu, Yi-Kai},
    title = {Direct fidelity estimation from few {Pauli} measurements},
    journal = {Physical Review Letters},
    volume = {106},
    number = {23},
    pages = {230501},
    doi = {10.1103/PhysRevLett.106.230501},
    year = {2011}
}

@inproceedings{ALL22,
    author = {Anshu, Anurag and Landau, Zeph and Liu, Yunchao},
    title = {Distributed quantum inner product estimation},
    booktitle = {Proceedings of the 54th Annual ACM SIGACT Symposium on Theory of Computing},
    pages = {44--51},
    doi = {10.1145/3519935.3519974},
    year = {2022}
}

@article{Wat09,
    author = {Watrous, John},
    title = {Zero-knowledge against quantum attacks},
    journal = {SIAM Journal on Computing},
    volume = {39},
    number = {1},
    pages = {25--58},
    doi = {10.1137/060670997},
    year = {2009}
}

@incollection{Nis25,
    author = {Nishimura, Harumichi},
    title = {A survey: {SWAP} test and its applications to quantum complexity theory},
    editor = {Minato, Shin-ichi and Uno, Takeaki and Yasuda, Norihito and Horiyama, Takashi and Kawarabayashi, Ken-ichi and Yamashita, Shigeru and Ono, Hirotaka},
    booktitle = {Algorithmic Foundations for Social Advancement: Recent Progress on Theory and Practice},
    volume = {},
    number = {},
    pages = {243--261},
    doi = {10.1007/978-981-96-0668-9_16},
    publisher = {Springer},
    series = {},
    year = {2025}
}

@article{WZC+23,
    author = {Wang, Qisheng and Zhang, Zhicheng and Chen, Kean and Guan, Ji and Fang, Wang and Liu, Junyi and Ying, Mingsheng},
    title = {Quantum Algorithm for Fidelity Estimation},
    journal = {IEEE Transactions on Information Theory},
    volume = {69},
    number = {1},
    pages = {273--282},
    doi = {10.1109/TIT.2022.3203985},
    year = {2023}
}

@article{WGL+24,
    author = {Wang, Qisheng and Guan, Ji and Liu, Junyi and Zhang, Zhicheng and Ying, Mingsheng},
    title = {New Quantum Algorithms for Computing Quantum Entropies and Distances},
    journal = {IEEE Transactions on Information Theory},
    volume = {70},
    number = {8},
    pages = {5653--5680},
    doi = {10.1109/TIT.2024.3399014},
    year = {2024}
}

@misc{GP22,
    author = {{Gily\'{e}n}, Andr\'{a}s and Poremba, Alexander},
    title = {Improved quantum algorithms for fidelity estimation},
    eprint = {2203.15993},
    howpublished = {ArXiv preprints},
    year = {2022}
}

@inproceedings{WZ24,
    author = {Wang, Qisheng and Zhang, Zhicheng},
    title = {Sample-optimal quantum estimators for pure-state trace distance and fidelity via samplizer},
    booktitle = {Proceedings of the 53rd International Colloquium on Automata, Languages, and Programming},
    pages = {154:1-154:21},
    doi = {10.4230/LIPIcs.ICALP.2026.154},
    year = {2026}
}

@article{LWWZ25,
    author = {Liu, Nana and Wang, Qisheng and Wilde, Mark M. and Zhang, Zhicheng},
    title = {Quantum algorithm for matrix geometric means},
    journal = {npj Quantum Information},
    volume = {11},
    number = {},
    pages = {101},
    doi = {10.1038/s41534-025-00973-7},
    year = {2025}
}

@article{LC19,
  doi = {10.22331/q-2019-07-12-163},
  url = {https://doi.org/10.22331/q-2019-07-12-163},
  title = {Hamiltonian {S}imulation by {Q}ubitization},
  author = {Low, Guang Hao and Chuang, Isaac L.},
  journal = {{Quantum}},
  issn = {2521-327X},
  volume = {3},
  pages = {163},
  year = {2019}
}

@article{KMY09,
    author = {Kobayashi, Hirotada and Matsumoto, Keiji and Yamakami, Tomoyuki},
    title = {Quantum {Merlin-Arthur} proof systems: are multiple {Merlins} more helpful to {Arthur}?},
    journal = {Chicago Journal of Theoretical Computer Science},
    volume = {2009},
    number = {},
    pages = {3},
    doi = {10.4086/cjtcs.2009.003},
    year = {2009}
}

@article{SH21,
    author = {Subramanian, Sathyawageeswar and Hsieh, Min-Hsiu},
    title = {Quantum algorithm for estimating $\alpha$-Renyi entropies of quantum states},
    journal = {Physical Review A},
    volume = {104},
    number = {2},
    pages = {022428},
    doi = {10.1103/PhysRevA.104.022428},
    year = {2021}
}

@misc{GHS21,
    author = {Gur, Tom and Hsieh, Min-Hsiu and Subramanian, Sathyawageeswar},
    title = {Sublinear quantum algorithms for estimating von {Neumann} entropy},
    howpublished = {ArXiv preprints},
    eprint = {2111.11139},
    year = {2021}
}

@article{WZL24,
    author = {Wang, Xinzhao and Zhang, Shengyu and Li, Tongyang},
    title = {A quantum algorithm framework for discrete probability distributions with applications to {R\'{e}nyi} entropy estimation},
    journal = {IEEE Transactions on Information Theory},
    volume = {70},
    number = {5},
    pages = {3399--3426},
    doi = {10.1109/TIT.2024.3382037},
    year = {2024}
}

@article{LW25,
    author = {Liu, Yupan and Wang, Qisheng},
    title = {On estimating the trace of quantum state powers},
    journal = {IEEE Transactions on Information Theory},
    volume = {},
    number = {},
    pages = {},
    doi = {10.1109/TIT.2026.3683891},
    year = {2026}
}

@article{WZ24b,
    author = {Wang, Qisheng and Zhang, Zhicheng},
    title = {Fast quantum algorithms for trace distance estimation},
    journal = {IEEE Transactions on Information Theory},
    volume = {70},
    number = {4},
    pages = {2720--2733},
    doi = {10.1109/TIT.2023.3321121},
    year = {2024}
}

@inproceedings{Bel19,
    author = {Belovs, Aleksandrs},
    booktitle = {Proceedings of the 27th Annual European Symposium on Algorithms},
    title = {Quantum algorithms for classical probability distributions},
    pages = {16:1--16:11},
    doi = {10.4230/LIPIcs.ESA.2019.16},
    year = {2019}
}

@article{RASW23,
    author = {Rethinasamy, Soorya and Agarwal, Rochisha and Sharma, Kunal and Wilde, Mark M.},
    title = {Estimating distinguishability measures on quantum computers},
    journal = {Physical Review A},
    volume = {108},
    number = {1},
    pages = {012409},
    doi = {10.1103/PhysRevA.108.012409},
    year = {2023}
}

@article{Uhl76,
    author = {Uhlmann, Armin},
    title = {The “transition probability” in the state space of a *-algebra},
    journal = {Reports on Mathematical Physics},
    volume = {9},
    number = {2},
    pages = {273--279},
    doi = {10.1016/0034-4877(76)90060-4},
    year = {1976}
}

@article{Joz94,
    author = {Jozsa, Richard},
    title = {Fidelity for mixed quantum states},
    journal = {Journal of Modern Optics},
    volume = {41},
    number = {12},
    pages = {2315--2323},
    doi = {10.1080/09500349414552171},
    year = {1994}
}

@book{Wil13,
    author = {Wilde, Mark M.},
    title = {Quantum Information Theory},
    publisher = {Cambridge University Press},
    doi = {10.1017/9781316809976},
    year = {2013}
}

@book{Wat18,
    author = {Watrous, John},
    title = {The Theory of Quantum Information},
    publisher = {Cambridge University Press},
    doi = {10.1017/9781316848142},
    year = {2018}
}

@book{Hay17,
    author = {Hayashi, Masahito},
    title = {Quantum Information Theory: Mathematical Foundation},
    publisher = {Springer},
    doi = {10.1007/978-3-662-49725-8},
    year = {2017}
}

@inproceedings{FW25c,
    author = {Fang, Wang and Wang, Qisheng},
    booktitle = {Proceedings of the 33rd Annual European Symposium on Algorithms},
    title = {Optimal quantum algorithm for estimating fidelity to a pure state},
    pages = {4:1--4:12},
    doi = {10.4230/LIPIcs.ESA.2025.4},
    year = {2025}
}

@misc{UNWT25,
    title = {Quantum algorithms for {Uhlmann} transformation},
    author = {Utsumi, Takeru and Nakata, Yoshifumi and Wang, Qisheng and Takagi, Ryuji},
    howpublished = {ArXiv preprints},
    eprint = {2509.03619},
    year = {2025}
}

@article{OW21,
    author = {O'Donnell, Ryan and Wright, John},
    title = {Quantum spectrum testing},
    journal = {Communications in Mathematical Physics},
    volume = {387},
    number = {1},
    pages = {1--75},
    doi = {10.1007/s00220-021-04180-1},
    year = {2021}
}

@article{HHJ+17,
    author = {Haah, Jeongwan and Harrow, Aram W. and Ji, Zhengfeng and Wu, Xiaodi and Yu, Nengkun},
    title = {Sample-optimal tomography of quantum states},
    journal = {IEEE Transactions on Information Theory},
    volume = {63},
    number = {9},
    pages = {5628--5641},
    doi = {10.1109/TIT.2017.2719044},
    year = {2017}
}

@inproceedings{OW16,
    author = {O'Donnell, Ryan and Wright, John},
    title = {Efficient quantum tomography},
    booktitle = {Proceedings of the 48th Annual ACM Symposium on Theory of Computing},
    pages = {899--912},
    doi = {10.1145/2897518.2897544},
    year = {2016}
}

@article{AISW20,
    author = {Acharya, Jayadev and Issa, Ibrahim and Shende, Nirmal V. and Wagner, Aaron B.},
    title = {Estimating Quantum Entropy},
    journal = {IEEE Journal on Selected Areas in Information Theory},
    volume = {1},
    number = {2},
    pages = {454--468},
    doi = {10.1109/JSAIT.2020.3015235},
    year = {2020}
}

@inproceedings{BOW19,
    author = {B{\u{a}}descu, Costin and O'Donnell, Ryan and Wright, John},
    title = {Quantum state certification},
    booktitle = {Proceedings of the 51st Annual ACM SIGACT Symposium on Theory of Computing},
    pages = {503--514},
    doi = {10.1145/3313276.3316344},
    year = {2019}
}

@inproceedings{AS25,
    author = {Arunachalam, Srinivasan and Schatzki, Louis},
    title = {Generalized inner product estimation with limited quantum communication},
    booktitle = {Proceedings of the 42nd International Symposium on Theoretical Aspects of Computer Science},
    pages = {11:1--11:17},
    doi = {10.4230/LIPIcs.STACS.2025.11},
    year = {2025}
}

@misc{ZWY+25,
    title = {Distributed quantum inner product estimation with low-depth circuits},
    author = {Zheng, Congcong and Wang, Kun and Yu, Xutao and Xu, Ping and Zhang, Zaichen},
    howpublished = {ArXiv preprints},
    eprint = {2506.19574},
    year = {2025}
}

@article{WZ25a,
    author = {Wang, Qisheng and Zhang, Zhicheng},
    title = {Quantum lower bounds by sample-to-query lifting},
    journal = {SIAM Journal on Computing},
    volume = {54},
    number = {5},
    pages = {1294--1334},
    doi = {10.1137/24M1638616},
    year = {2025},
}

@article{WZ25b,
    author = {Wang, Qisheng and Zhang, Zhicheng},
    title = {Time-efficient quantum entropy estimator via samplizer},
    journal = {IEEE Transactions on Information Theory},
    volume = {71},
    number = {12},
    pages = {9569--9599},
    doi = {10.1109/TIT.2025.3576137},
    year = {2025},
}

@misc{TWZ25,
    title = {Conjugate queries can help},
    author = {Tang, Ewin and Wright, John and Zhandry, Mark},
    howpublished = {ArXiv preprints},
    eprint = {2510.07622},
    year = {2025}
}

@misc{GHYZ24,
    title = {On the sample complexity of purity and inner product estimation},
    author = {Gong, Weiyuan and Haferkamp, Jonas and Ye, Qi and Zhang, Zhihan},
    howpublished = {ArXiv preprints},
    eprint = {2410.12712},
    year = {2024}
}

@misc{CWZ25,
    title  = {A list of complexity bounds for property testing by quantum sample-to-query lifting},
    author = {Chen, Kean and Wang, Qisheng and Zhang, Zhicheng},
    note   = {ArXiv preprint},
    eprint = {2512.01971},
    year   = {2025}
}

@phdthesis{Har05,
    Author = {Harrow, Aram W.},
    Title = {Applications of coherent classical communication and the {S}chur transform to quantum information theory},
    school = {Massachusetts Institute of Technology},
    Year = {2005},
    Url = {http://hdl.handle.net/1721.1/34973},
}

@article{BCH06,
    title = {Efficient Quantum Circuits for {S}chur and {C}lebsch-{G}ordan Transforms},
    author = {Bacon, Dave and Chuang, Isaac L. and Harrow, Aram W.},
    journal = {Physical Review Letters},
    volume = {97},
    issue = {17},
    pages = {170502},
    year = {2006},
    doi = {10.1103/PhysRevLett.97.170502},
}

@inproceedings{BCH07,
    author = {B{\u{a}}descu, Costin and O'Donnell, Ryan and Wright, John},
    title = {The quantum {Schur} and {Clebsch-Gordan} transforms: {I.} efficient qudit circuits},
    booktitle = {Proceedings of the 18th Annual ACM-SIAM Symposium on Discrete Algorithms},
    pages = {1235--1244},
    url = {https://dl.acm.org/doi/abs/10.5555/1283383.1283516},
    year = {2007}
}

@article{Kro19,
    title = {An efficient high dimensional quantum {S}chur transform},
    author = {Krovi, Hari},
    journal = {Quantum},
    volume = {3},
    pages = {122},
    year = {2019},
    doi = {10.22331/q-2019-02-14-122},
}

@misc{BFG+25,
    title  = {High-dimensional quantum {S}chur transforms},
    author = {Burchardt, Adam and Fei, Jiani and Grinko, Dmitry and Larocca, Martin and Ozols, Maris and Timmerman, Sydney and Visnevskyi, Vladyslav},
    note   = {ArXiv preprint},
    eprint = {2509.22640},
    year   = {2025}
}

@article{Hel67,
    author = {Helstrom, Carl W.},
    title = {Detection theory and quantum mechanics},
    journal = {Information and Control},
    volume = {10},
    number = {3},
    pages = {254--291},
    doi = {10.1016/S0019-9958(67)90302-6},
    year = {1967}
}

@article{Hol73,
    author = {Holevo, Alexander S.},
    title = {Statistical decision theory for quantum systems},
    journal = {Journal of Multivariate Analysis},
    volume = {3},
    number = {4},
    pages = {337--394},
    doi = {10.1016/0047-259X(73)90028-6},
    year = {1973}
}

@article{FvdG99,
    author = {Fuchs, Christopher A. and van de Graaf, Jeroen},
    title = {Cryptographic distinguishability measures for quantum-mechanical states},
    journal = {IEEE Transactions on Information Theory},
    volume = {45},
    number = {4},
    pages = {1216--1227},
    doi = {10.1109/18.761271},
    year = {1999}
}

@inproceedings{HH00,
    author = {Hales, Lisa and Hallgren, Sean},
    title = {An improved quantum Fourier transform algorithm and applications},
    booktitle = {Proceedings of the 41st Annual Symposium on Foundations of Computer Science},
    pages = {515--525},
    doi = {10.1109/SFCS.2000.892139},
    year = {2000}
}

@misc{Kit95,
    author = {Kitaev, A. Yu.},
    title = {Quantum measurements and the {Abelian} stabilizer problem},
    howpublished = {ArXiv e-prints},
    eprint = {quant-ph/9511026},
    year = {1995}
}

@article{Sho97,
    author = {Shor, Peter W.},
    title = {Polynomial-time algorithms for prime factorization and discrete logarithms on a quantum computer},
    journal = {SIAM Journal on Computing},
    volume = {26},
    number = {5},
    pages = {1484--1509},
    doi = {10.1137/S0097539795293172},
    year = {1997}
}

@inproceedings{PSW25,
    author = {Pelecanos, Angelos and Spilecki, Jack and Wright, John},
    title = {The debiased {Keyl's} algorithm: A new unbiased estimator for full state tomography},
    booktitle = {Proceedings of the 58th Annual ACM Symposium on Theory of Computing},
    pages = {1266--1277},
    doi = {10.1145/3798129.3800837},
    year = {2026}
}

@misc{SSW25,
    author = {Scharnhorst, Thilo and Spilecki, Jack and Wright, John},
    title = {Optimal lower bounds for quantum state tomography},
    howpublished = {ArXiv preprints},
    eprint = {2510.07699},
    year = {2025}
}

@misc{PSTW25,
    author = {Pelecanos, Angelos and Spilecki, Jack and Tang, Ewin and Wright, John},
    title = {Mixed state tomography reduces to pure state tomography},
    howpublished = {ArXiv preprints},
    eprint = {2511.15806},
    year = {2025}
}

@inproceedings{CW25,
    author = {Chen, Kean and Wang, Qisheng},
    title = {Improved Sample Upper and Lower Bounds for Trace Estimation of Quantum State Powers},
    booktitle = {Proceedings of the 38th Annual Conference on Learning Theory},
    pages = {1008--1028},
    doi = {},
    url = {https://proceedings.mlr.press/v291/chen25d.html},
    year = {2025}
}

@inproceedings{CHW07,
    author = {Childs, Andrew M. and Harrow, Aram W. and Wocjan, Pawe{\l}},
    title = {Weak {Fourier-Schur} sampling, the hidden subgroup problem, and the quantum collision problem},
    booktitle = {Proceedings of the 24th Annual Symposium on Theoretical Aspects of Computer Science},
    pages = {598-609},
    doi = {10.1007/978-3-540-70918-3_51},
    year = {2007}
}

@article{Hay25,
    author = {Hayashi, Masahito},
    title = {Measuring quantum relative entropy with finite-size effect},
    journal = {Quantum},
    volume = {9},
    number = {},
    pages = {1725},
    doi = {10.22331/q-2025-05-05-1725},
    year = {2025}
}

@article{GKP+25,
    author = {Go, Byeongseon and Kwon, Hyukjoon and Park, Siheon and Patel, Dhrumil and Wilde, Mark M.},
    title = {Sample-based {Hamiltonian} and {Lindbladian} simulation: Non-asymptotic analysis of sample complexity},
    journal = {Quantum Science and Technology},
    volume = {10},
    number = {4},
    pages = {045058},
    doi = {10.1088/2058-9565/ae075b},
    year = {2025}
}

@article{CWZ24,
    author = {Chen, Kean and Wang, Qisheng and Zhang, Zhicheng},
    title = {Local test for unitarily invariant properties of bipartite quantum states},
    journal = {IEEE Transactions on Information Theory},
    volume = {},
    number = {},
    pages = {},
    doi = {10.1109/TIT.2026.3697790},
    year = {2026}
}

@inproceedings{SW22,
  title={Testing matrix product states},
  author={Soleimanifar, Mehdi and Wright, John},
  booktitle={Proceedings of the 2022 Annual ACM-SIAM Symposium on Discrete Algorithms},
  pages={1679--1701},
  doi={10.1137/1.9781611977073.68},
  year={2022},
}

@inproceedings{Kre21,
  title={Quantum pseudorandomness and classical complexity},
  author={Kretschmer, William},
  booktitle={Proceedings of the 16th Conference on the Theory of Quantum Computation, Communication and Cryptography},
  pages={2:1--2:20},
  doi={10.4230/LIPIcs.TQC.2021.2},
  year={2021},
}

@inproceedings{GZ25,
    author = {Goldin, Eli and Zhandry, Mark},
    title = {Translating between the common {Haar} random state model and the unitary model},
    booktitle = {Proceedings of the 45th Annual International Cryptology Conference},
    pages = {269--300},
    doi = {10.1007/978-3-032-01878-6_9},
    year = {2025}
}

@misc{TW25,
    author = {Tang, Ewin and Wright, John},
    title = {Are controlled unitaries helpful?},
    howpublished = {ArXiv preprints},
    eprint = {2508.00055},
    year = {2025}
}

@article{LMR14,
	author = {Lloyd, Seth and Mohseni, Masoud and Rebentrost, Patrick},
	journal = {Nature Physics},
	title = {Quantum principal component analysis},
	volume = {10},
	number = {9},
	pages = {631-633},
	doi = {10.1038/nphys3029},
	year = {2014}
}

@article{ACM+07,
  title     = {Discriminating states: The quantum {Chernoff} bound},
  author    = {Audenaert, K. M. R. and Calsamiglia, J. and Mu{\~n}oz-Tapia, R. and Bagan, E. and Masanes, Ll. and Acin, A. and Verstraete, F.},
  journal   = {Physical Review Letters},
  volume    = {98},
  number    = {16},
  pages     = {160501},
  year      = {2007},
  doi       = {10.1103/PhysRevLett.98.160501},
  publisher = {APS}
}

@article{Kar05,
	author = {Kargin, Vladislav},
	journal = {Annals of Statistics},
	title = {On the {Chernoff} bound for efficiency of quantum hypothesis testing},
	volume = {33},
	number = {2},
	pages = {959--976},
	doi = {10.1214/009053604000001219},
	year = {2005}
}

@article{CDL+25,
	author = {Cheng, Hao-Chung and Datta, Nilanjana and Liu, Nana and Nuradha, Theshani and Salzmann, Robert and Wilde, Mark M.},
	journal = {npj Quantum Information},
	title = {An invitation to the sample complexity of quantum hypothesis testing},
	volume = {11},
	number = {},
	pages = {94},
	doi = {10.1038/s41534-025-00980-8},
	year = {2025}
}

@misc{FWZ25,
    author = {Fang, Wang and Wang, Qisheng and Zhang, Zhicheng},
    title = {Optimal quantum estimators for pure-state closeness},
    url = {https://aqis-conf.org/2025/wp-content/uploads/2025/08/NewestBookletTalks.pdf},
    year = {2025},
    note = {Contributed talk at the 25th Asian Quantum Information Science Conference},
}

@article{dSLCP11,
	author = {da Silva, Marcus P. and Landon-Cardinal, Olivier and Poulin, David},
	journal = {Physical Review Letters},
	title = {An invitation to the sample complexity of quantum hypothesis testing},
	volume = {107},
	number = {21},
	pages = {210404},
	doi = {10.1103/PhysRevLett.107.210404},
	year = {2011}
}

@article{PLM18,
	author = {Pallister, Sam and Linden, Noah and Montanaro, Ashley},
	journal = {Physical Review Letters},
	title = {Optimal verification of entangled states with local measurements},
	volume = {120},
	number = {17},
	pages = {170502},
	doi = {10.1103/PhysRevLett.120.170502},
	year = {2018}
}

@article{HKP20,
	author = {Huang, Hsin-Yuan and Kueng, Richard and Preskill, John},
	journal = {Nature Physics},
	title = {Predicting many properties of a quantum system from very few measurements},
	volume = {16},
	number = {},
	pages = {1050--1057},
	doi = {10.1038/s41567-020-0932-7},
	year = {2020}
}

@misc{Wan26,
    title = {Estimating Fidelity to a Reference Quantum State},
    author = {Wang, Qisheng},
    howpublished = {ArXiv preprints},
    eprint = {2606.26034},
    year = {2026}
}

@misc{LT26,
    title = {Random dimension reduction and learning symmetric properties of quantum states},
    author = {Lowe, Angus and Tan, Xinyu},
    howpublished = {ArXiv preprints},
    eprint = {2606.23592},
    year = {2026}
}

\end{document}